\documentclass[a4paper,USenglish]{article}

\usepackage[utf8]{inputenc}
\usepackage[T1]{fontenc}
\usepackage{xcolor}

\usepackage{amsmath,amssymb,mathrsfs}
\usepackage{graphicx}
\usepackage{subcaption}
\usepackage{enumitem}

\usepackage{todonotes}
\usepackage{dsfont}

\usepackage{booktabs}
\usepackage{tabularx}
\usepackage{multirow}
\usepackage{a4wide}

\usepackage[ruled,vlined,linesnumbered]{algorithm2e}

\DeclareMathOperator{\opt}{OPT}

\DeclareMathOperator{\cost}{Cost}
\DeclareMathOperator{\mcc}{MCC}
\DeclareMathOperator{\rel}{rel}
\DeclareMathOperator{\dom}{reach}
\DeclareMathOperator{\sign}{sign}

\newcommand{\uug}[1]{\ensuremath{U_{#1}}}
\newcommand{\N}{\mathds{N}}
\newcommand{\Z}{\mathds{Z}}

\newcommand{\lin}{linearizable}
\newcommand{\Zlin}[1]{#1-$\Z$-\lin}

\newcommand{\rsg}{G_\ell^{\rel}} 
\newcommand{\rw}{w_\ell^{\rel}} 
\newcommand{\sdom}{\mathcal{S}_{G,\ell}^{\dom}} 

\SetKwInOut{Input}{Input}\SetKwInOut{Output}{Output}
\SetKwProg{Fn}{Function}{}{}
\SetKw{Continue}{continue}
\DontPrintSemicolon

\let\oldnl\nl  
\newcommand{\nonl}{\renewcommand{\nl}{\let\nl\oldnl}}  

\usepackage{tikz}
\usetikzlibrary{backgrounds}
\tikzset{vertex/.style = {shape=circle,draw,minimum size=3em}}
\tikzset{edge/.style = {->,> = latex'}}
\tikzset{> = stealth, shorten > = 1pt, auto, semithick}

\usepackage[numbers,sort,sectionbib]{natbib}

\usepackage[pagebackref,menucolor=orange!40!black,filecolor=magenta!40!black,urlcolor=blue!40!black,linkcolor=red!40!black,citecolor=green!40!black,colorlinks]{hyperref}
\usepackage[nameinlink,capitalize]{cleveref}

\usepackage{amsthm}
\newtheorem{theorem}{Theorem}[section]
\newtheorem{corollary}[theorem]{Corollary}
\newtheorem{lemma}[theorem]{Lemma}
\newtheorem{observation}[theorem]{Observation}
\newtheorem{proposition}[theorem]{Proposition}

\theoremstyle{definition}

\newtheorem{definition}[theorem]{Definition}

\newtheorem{reduction}[theorem]{Reduction}
\newtheorem{reductionrule}[theorem]{Reduction Rule}
\newtheorem{remark}[theorem]{Remark}
\newtheorem{claim}{Claim}

\crefname{reductionrule}{Reduction Rule}{Reduction Rules}
\crefname{reduction}{Reduction}{Reduction}
\crefname{observation}{Observation}{Observations}

\newcommand{\problemdef}[4]{
    \begin{center}
        \begin{minipage}{0.98\textwidth}
            \textsc{#2}
            
            \setlength{\tabcolsep}{2pt}
            \begin{tabularx}{\textwidth}{@{}lX@{}}
                \textbf{Input:}     & #3 \\
                \textbf{#1:}        & #4
            \end{tabularx}
        \end{minipage}
    \end{center}
}
\newcommand{\decProb}[4][Question]{\problemdef{#1}{#2}{#3}{#4}}
\newcommand{\optProb}[4][Task]{\problemdef{#1}{#2}{#3}{#4}}

\title{\Large \bf Parameterized Complexity of Min-Power Asymmetric Connectivity}
\author{Matthias Bentert \and Roman Haag \and Christian Hofer \and Tomohiro Koana \and Andr\'e Nichterlein}

\date{Algorithmics and Computational Complexity, Faculty IV, TU Berlin, Germany, \texttt{\{matthias.bentert,tomohiro.koana,andre.nichterlein\}@tu-berlin.de, \{roman.haag,hofer\}@campus.tu-berlin.de}}

\begin{document}

\pagestyle{plain}

\maketitle 

\begin{abstract}
We investigate parameterized algorithms for the NP-hard problem \textsc{Min-Power Asymmetric Connectivity (MinPAC)} that has applications in wireless sensor networks.
Given a directed arc-weighted graph, MinPAC asks for a strongly connected spanning subgraph minimizing the summed vertex costs.
Here, the cost of each vertex is the weight of its heaviest outgoing arc in the chosen subgraph. 
We present linear-time algorithms for the cases where the number of strongly connected components in a so-called obligatory subgraph or the feedback edge number in the underlying undirected graph is constant. 
Complementing these results, we prove that the problem is W[2]-hard with respect to the solution cost, even on restricted graphs with one feedback arc and binary arc weights.
\end{abstract}

\section{Introduction}

In wireless ad-hoc networks, nodes equipped with limited power supply transmit data using a multi-hop path. 
We study the problem of minimizing the overall power consumption while maintaining full network connectivity, that is, each node can send messages to each other node using some (multi-hop) route through the network.
Formally, we study the following optimization problem. 

\optProb{\textsc{Min-Power Asymmetric Connectivity (MinPAC)}}
{A strongly connected directed graph~$G$ and a weight (cost) function~$w\colon A(G) \to \N$.}
{Find a strongly connected spanning subgraph~$H$ of~$G$ minimizing
$$\sum_{v \in V} \max_{vu \in A(H)} w(vu).$$}
\vspace{-0.8cm}
\paragraph{Related work.} 
This problem was initially formalized and shown to be NP-complete by Chen and Huang \cite{chen1989strongly}. 
Since then, there have been numerous publications on polynomial-time approximation algorithms (an asymptotically optimal~$O(\log n)$ approximation \cite{Calinescu2003}, a constant approximation factor with symmetric arc weights \citep{chen1989strongly,Calinescu2013approximate}, approximation algorithms for special cases \cite{Carmi2007,chen2005dual,Calinescu2016}), and hardness results for special cases~\cite{clementi2004power,Carmi2007}. 
To the best of our knowledge, the parameterized complexity of \textsc{MinPAC} has not been investigated yet. 

In previous work, we investigated the parameterized complexity of the symmetric version of our problem~\cite{BBNN17}; the difference to MinPAC is that an undirected graph is given and for every undirected edge in the solution subgraph~$H$ both endpoints pay at least the weight of the edge. 
The asymmetric case turns out to be more involved on a technical level.
However, comparable results (to the symmetric case) are achievable.
\paragraph{Our contributions.}
We show algorithmic results for grid-like and tree-like input graphs as well as parameterized hardness for very restricted cases. 
\Cref{tab:results} summarizes our results. 
We discuss the different parameters subsequently.
\begin{table}[t]
    \caption{
        Overview of our results, using the following notation:  
        $n$---number of vertices, 
        $m$---number of arcs, 
        $c$---number of strongly connected components in the obligatory subgraph (see \Cref{sec:fpt}), 
        $q$---number of different arc weights, 
        $x$---size of a minimum vertex cover in the underlying undirected graph (see \Cref{sec:vcn}), 
        $g$---size of a minimum feedback edge set of the underlying undirected graph (see \Cref{sec:fen}),  
        $h$---size of a minimum feedback arc set (see \Cref{section:hardness}), 
        \textsc{PAC}---the decision version of \textsc{MinPAC} asking for a solution of cost at most~$k$.
    } 
    \label{tab:results}
    \setlength{\tabcolsep}{4pt}
    \begin{tabularx}{\textwidth}{ l X l }
        \toprule
        & result & reference \\ \midrule
        \multirow{2}{*}{\rotatebox{90}{S.~\ref{sec:fpt}}}   & Dynamic programming solving \textsc{MinPAC} in                & \multirow{2}{*}{\Cref{thm:SCC-Algo}} \\  
                                                & $O(c^2 \cdot 2^c \cdot n + m + 4^c \cdot c^{2c - 3/2})$ time.     & \\  
        \midrule
        \multirow{3}{*}{\rotatebox{90}{S.~\ref{sec:vcn}}} & An~$O(\min(\{ x, \log_q n \}) \cdot n + m)$-time data reduction resulting in an equivalent \textsc{MinPAC} instance with at most~$(q + 1) ^{2 x} + x$ vertices.& \multirow{2}{*}{\Cref{thm:VCKern}} \\
                                                & An exponential-size kernel with respect to~$x+q$.     & \cref{cor:expKernelVCq}\\ 
        \midrule
        \multirow{3}{*}{\rotatebox{90}{S.~\ref{sec:fen}}} & Linear-time data reduction resulting in an equivalent \textsc{MinPAC} instance with at most~$20g-20$ vertices and~$42g - 42$ arcs.
        & \multirow{2}{*}{\Cref{theorem:kernelization}} \\ 
                                                & A polynomial-size kernel with respect to~$g$.     & \cref{cor:poly-kernel-fes} \\  
        \midrule
        \multirow{4}{*}{\rotatebox{90}{\cref{section:hardness}}} & \textsc{PAC} is NP-hard for any~$h \geq 1$. & \multirow{4}{*}{\Cref{theo:hardness}} \\
        & \textsc{PAC} is W[2]-hard parameterized by~$k$, even if the arcs have only cost zero or one and~$h = 1$. & \\
                                                        & \textsc{PAC} is not solvable in~$2^{o(n)}$ time (assuming ETH).& \\
        \bottomrule
    \end{tabularx}
\end{table}

It is known that the alignment of nodes in some regular grid-like patterns is optimal to fully cover a plane. 
In such cases, we can assume that the \emph{obligatory arcs}, arcs that are in any optimal solution, induce a small number~$c$ of strongly connected components as there are many arcs of minimum weight. 
In \Cref{sec:fpt}, we define obligatory arcs, discuss the connection to grid-like graphs, and present an algorithm that solves \textsc{MinPAC} in linear time when~$c$ is a constant. 

In \cref{sec:vcn}, we study \textsc{MinPAC} parameterized by the number of different arc weights and the vertex cover number. 
For this combined parameter, we present an exponential-size kernel.

In \Cref{sec:fen}, we describe a linear-time algorithm which reduces any input instance to an equivalent instance with at most~$20g - 20$ vertices and~$42g - 42$ arcs, where~$g$~is the feedback edge number of the underlying undirected graph. 
The parameter is also motivated by real world applications in which the feedback edge number is small; for instance, sensor networks along waterways (including canals) are expected to have a small number of feedback edges.
It follows from our result hat the problem can be solved in polynomial time for~$g \in O(\log n)$, that is, for very tree-like input graphs. 
In terms of parameterized complexity, this gives us a partial (weights left unbounded) kernelization of \textsc{MinPAC} with respect to the feedback edge number.
Using an existing weight-shrinking technique~\cite{BBFNN}, we also provide a ``full'' polynomial-size kernel with respect to the feedback edge number.

Finally, in \Cref{section:hardness} we derive hardness results for \textsc{PAC}, the decision version of \textsc{MinPAC}. 
We show that even if the input graph has only binary weights and is almost a DAG (a directed acyclic graph with one additional arc), \textsc{PAC} parameterized by the solution cost is W[2]-hard. 
This is in sharp contrast to the FPT result for the parameter feedback edge number.

\paragraph{Preliminaries.} 
For~$a \in \N$, we abbreviate~$\{1, \dots, a\}$ by~$[a]$. 
Throughout this work, we assume that a graph is directed unless stated otherwise.  
For a graph~$G = (V, A)$, we write~$V(G)$ to denote~$V$ and~$A(G)$ to denote~$A$. 
We abbreviate arcs~$(u,v) \in A$ by~$uv$.
We denote by~$G[V']$ the subgraph induced by~$V'\subseteq V(G)$.
We use~$G + vu$ to denote~$(V(G) \cup \{ v, u \}, A(G) \cup \{ vu \})$ and~$G - vu$ to denote~$(V(G), A(G) \setminus \{ vu \})$.
For a vertex~$v \in V(G)$, we write~$N^+_G(v) = \{ u \mid vu \in A \}$ and~$N^-_G(v) = \{ u \mid uv \in A \}$
to denote the out- and in-neighborhood of~$v$. 
We define the degree of~$v$ as~$\deg_G(v) = |N^+_G(v) \mathbin{\cup} N^-_G(v)|$. 
We say that~$S \subseteq V(G)$ is a strongly connected component if there exists a path from each vertex~$u \in S$ to every other vertex~$v \in S$ in~$G[S]$. 
We write~$\mathcal{S}_G$ to denote the set of strongly connected components. 
We use~$\uug{G}$ to denote the underlying undirected graph of~$G$.
We denote the optimal cost of an instance of \textsc{MinPAC}~$I$ by~$\opt(I)$. 
The cost of a vertex subset~$V' \subseteq V(G)$ in a solution with arcs~$A' \subseteq A(G)$ is denoted by~$\cost(V', A', w) = \sum_{v \in V'} \max_{vu \in A'} w(vu)$. 
For ease of notation, we write~$w(vu) = \infty$ if~$vu \not\in A$.

A parameterized problem~$\Pi$ is a set of pairs~$(I, k)$, where~$I$ denotes the problem instance and~$k$ is the parameter. 
The problem~$\Pi$ is \emph{fixed-parameter tractable} (FPT) if there exists an algorithm solving any instance of~$\Pi$ in~$f(k) \cdot |I|^c$ time, where~$f$ is some computable function and~$c$ is some constant. 
A \emph{reduction to a problem kernel} is a polynomial-time algorithm that, given an instance~$(I, k)$ of~$\Pi$, returns an equivalent instance~$(I', k')$, such that~$|I'| + k' \le g(k)$ for some computable function~$g$.
Problem kernels are usually achieved by applying \emph{data reduction rules}.
Given an instance~$(I, k)$ for \textsc{MinPAC}, our data reduction rules compute in polynomial time a new instance~$(I', k')$ of \textsc{MinPAC} and a number~$d$.
We call a data reduction rule \emph{correct}, if~$\opt(I) = \opt(I') + d$.

\section{Parameterization by the number of strongly connected components induced by the obligatory arcs}
\label{sec:fpt}
In this section we present a fixed-parameter algorithm with respect to the number~$c$ of strongly connected components (SCCs) induced by \emph{obligatory arcs}---arcs that can be included into any optimal solution with no additional cost.
We find the obligatory arcs by means of lower bounds on costs paid by each vertex.

\begin{definition}
  A \emph{vertex lower bound} is a function~$\ell \colon V(G) \to \N$ such that for any optimal solution~$H$ and any vertex~$v \in V(G)$, it holds that
  \[
    \ell(v) \le \max_{vu \in A(H)} w(vu).
  \]
\end{definition}

Observe that each vertex~$v \in V(G)$ has at least one outgoing arc in any optimal solution.
Hence, the cost paid by~$v$ in any optimal solution is at least~$\min_{vu \in A(G)} w(vu)$.
Thus, $\ell(v) \ge \min_{vu \in A(G)} w(vu)$.
Moreover, if a vertex~$v$ has only one incoming arc~$uv$, then the cost for the vertex~$u$ is at least~$w(uv)$, and thus~$\ell(u) \ge w(uv)$.
Clearly, finding more effective but still efficiently computable vertex lower bounds is challenging on its own.

\begin{definition} \label{obligatoryGraph} 
    The \emph{obligatory subgraph}~$G_\ell$ induced by a vertex lower bound~$\ell$ for~$G$ is a subgraph~$(V(G), A_\ell)$, where~$A_\ell = \{vu \mid w(vu) \le \ell(v) \}.$
\end{definition}

It has been shown that sensors are optimally placed  for fully covering an area when sensors are deployed in a triangular lattice pattern \cite{Zhang2005} or a strip based-pattern \cite{Bai2006,IKB09}. 
In such cases, there are many arcs of minimum weight.
Taking these arcs usually suffices to (almost) achieve strong connectivity.
So even the obligatory subgraph induced by the trivial vertex lower bound described above yields a small number of SCCs.

Let~$\ell$ be a vertex lower bound for a graph~$G$. 
We denote the number of SCCs of the obligatory subgraph~$G_\ell$ by~$c = |\mathcal{S}_{G_\ell}|$.
The number~$c$ of (strongly) connected components in the obligatory subgraph has recently been used as parameter to obtain FPT results \cite{BBNN17,sorge2012new}.
In this section, we also provide an FPT result with respect to this parameter.
More specifically, we will present an algorithm for \textsc{MinPAC} that runs in~$O(2^c\cdot c^2 \cdot n + m + 4^c \cdot c^{2c - 3/2})$ time.
Our algorithm runs in three phases.
In the first phase, it shrinks the graph to a \emph{relevant} subgraph in which each vertex~$v$ has at most one arc towards each SCC that does not contain~$v$ (\Cref{scc:redalgo}).
In the second phase, it uses a dynamic programming algorithm to compute the minimum cost to connect each SCC to each subset of other SCCs (\Cref{scc:dp}).
In the last phase, it exhaustively tries all combinations of connecting SCCs to find an optimal solution (\Cref{alg:minpac}). 

\paragraph{Phase 1.} 
The following lemma specifies the conditions under which we can remove arcs. 
It plays a central role in this phase.
The basic idea herein is to remove, for each vertex~$v \in V(G)$ and each SCC~$S$, all but the cheapest arc from~$v$ to vertices in~$S$.

\begin{lemma}
\label{scc:arcred}
Let~$(G, w)$ be an instance of \textsc{MinPAC} and let~$\ell$ be a vertex lower bound.
Let~$S_v, S_u \in \mathcal{S}_{G_\ell}$ be two distinct SCCs and let~$v \in S_v$ and~$u, u' \in S_u$ be vertices of~$G$ with~$w(vu) \le w(vu')$. 
Then, it holds that~$\opt((G, w)) = \opt((G - vu', w))$.
\end{lemma}
{
\begin{proof}
  Observe that the removal of an arc does not decrease the cost of an optimal solution.
  Thus, $\opt((G, w)) \le \opt((G - vu', w))$.
  Let~$H$ be an optimal solution of~$(G, w)$.
  Suppose that~$H$ contained~$vu'$.
  The cost paid by~$v$ in~$H$ is then at least~$w(vu')$.
  Since~$u$ and~$u'$ both belong to~$S_u$, the subgraph~$H' = H + vu - vu'$ is a strongly connected spanning subgraph.
  Since~$w(vu) \le w(vu')$, it follows that~$H'$ is optimal. 
  Moreover,~$H'$~is also a solution of~$(G - vu', w)$, and thus~$\opt((G, w)) \ge \opt((G - vu', w))$.
\end{proof}
}
\Cref{scc:redalgo} exhaustively removes all arcs~$vu'$ which satisfy the preconditions of \Cref{scc:arcred}:
The algorithm iterates over each arc in~$G$ twice.
It finds a minimum-weight arc from each vertex to each SCC in the first iteration.
In the second iteration, it removes all but one minimum-weight arc that share the initial vertex and the SCC the terminal vertex belongs to.

We show subsequently that the resulting instance of \textsc{MinPAC} satisfies the properties listed in the next definition.

\begin{definition}
  \label{def:rel}
  Let~$(G, w)$ be an instance of \textsc{MinPAC} and let~$\ell$ be a lower bound.
  We say that a graph~$\rsg$ and a weight function~$\rw$ with~$V(G) = V(\rsg)$, $A(\rsg) \subseteq A(G)$, and~$\rw \colon A(\rsg) \to \N$ are \emph{a relevant subgraph} and \emph{relevant weight function} induced by~$\ell$, respectively, if they satisfy the following properties:
  \begin{enumerate}[align=left,label=(\roman*)]
    \setlength{\itemindent}{.5em}
    \setlength{\labelwidth}{1.5em}
    \item $\opt((G, w)) = \opt((\rsg, \rw))$.
    \item For any SCC~$S \in \mathcal{S}_{G_\ell}$, it holds that~$G[S] = \rsg[S]$.\label{def:rel2}
    \item For any SCC~$S \in \mathcal{S}_{G_\ell}$ and any vertex~$v \notin S$, it holds that~$|\{ vu \in A(\rsg) \mid u \in S \}| \le 1$. 
  \end{enumerate}
\end{definition}

Since it follows from property \ref{def:rel2} that~$\mathcal{S}_{G_\ell} = \mathcal{S}_{\rsg}$, we will use them interchangeably.

\begin{algorithm}[t]
  \caption{A reduction procedure for the first phase}
  \label{scc:redalgo}
  \SetKwFunction{FRed}{Reduction}
  \SetKwProg{Fn}{Function}{}{}
  \SetKw{null}{null}

  \Fn{\FRed{$G, w, \ell$}}{
    \tcp{$S_u \in \mathcal{S}_{G_\ell}$ denotes the SCC to which~$u$ belongs}
    \lForEach{$vu \in A(G)$}{$M(v, S_u) \leftarrow \null$} \label{alg:one:init}
    \ForEach{$vu \in A(G)$}{ \label{alg:one:iter1}
        \lIf{$M(v, S_u) = $ \null or~$w(vu) < w(M(v, S_u))$}{$M(v, S_u) \leftarrow vu$}
    }
    \ForEach{$vu \in A(G)$}{ 
      \lIf{$S_v = S_u$}{\Continue}
      \lIf{$vu \ne M(v, S_u)$}{remove~$vu$ from~$G$}
    } \label{alg:one:iter2}
    \Return $(G, w)$
  }
\end{algorithm}

\begin{lemma}
  \label{scc:ared}
  Let~$(G, w)$ be an instance of \textsc{MinPAC} and let~$\ell$ be a vertex lower bound.
  \Cref{scc:redalgo} computes in~$O(m)$ time a relevant subgraph~$\rsg$ and a relevant weight function~$\rw$ induced by~$\ell$.
\end{lemma}
{
\begin{proof}
  For any~$v \in V(G)$ and any~$S \in \mathcal{S}_{G_\ell}$, \Cref{scc:redalgo} sets~$M(v, S)$ to a minimum-weight arc from~$v$ to a vertex in~$S$ in the first iteration.
  In the second iteration, every arc~$vu$ is removed unless~$v$ and~$u$ belong to the same SCC or~$M(v, S_u)$ is set to~$vu$.
  Recall that the optimal cost remains the same after arc removals as shown in \Cref{scc:arcred}.
  Thus, \Cref{scc:redalgo} returns an instance of \textsc{MinPAC}~$(\rsg, \rw)$ that satisfies the aforementioned properties. 
  The algorithm spends~$O(m)$ time on initializing~$M$ (\Cref{alg:one:init}) and~$O(m)$ time on the iteration over the arcs (\Crefrange{alg:one:iter1}{alg:one:iter2}).
\end{proof}
}
\paragraph{Phase 2.} 
In this phase, we aim to compute an optimal set of arcs to connect each SCC to all other SCCs. 
We start with some notation.

\begin{definition}
  Let~$\rsg$ be a relevant subgraph.
  For any~$S \in \mathcal{S}_{G_\ell}$, we define the set of SCCs reachable from~$S$ via an arc as 
  \[
    \sdom(S) = \{ S' \in \mathcal{S}_{G_\ell} \setminus \{ S \} \mid \exists vu \in A(\rsg).\; v  \in S \wedge u \in S' \}.
  \]
\end{definition}

We say that an arc set~$B$ is a \emph{connector} if~$B$ connects some SCC~$S$ to some set~$\mathcal{T} \subseteq \sdom(S)$ of SCCs reachable from~$S$.
Then, our goal is to find a connector of minimum cost for each~$S \in \mathcal{S}_{G_\ell}$ and each subset~$\mathcal{T} \subseteq \sdom(S)$.
This allows us to compute an optimal solution with exhaustive search on connections between SCCs in the last phase.

\begin{definition} \label{def:mcc}
  Let~$(G, w)$ be an instance of \textsc{MinPAC} and let~$\ell$ be a vertex lower bound.
  A minimum-cost connector is a function~$\mcc \colon \mathcal{S}_{G_\ell} \times 2^{\mathcal{S}_{G_\ell}} \to 2^{A(\rsg)}$ such that for any~$S \in \mathcal{S}_{G_\ell}$ and any~$\mathcal{T} \subseteq \sdom(S)$ the following properties are satisfied:
  \begin{enumerate}%
    \item For any~$S' \in \mathcal{T}$, there exist vertices~$v \in S$ and~$u \in S'$ with~$vu \in \mcc(S, \mathcal{T})$.\label{def:mcc1}
    \item There is no subset~$B \subseteq A(\rsg)$ that satisfies the above property and that satisfies~$\cost(S, B, \rw) < \cost(S, \mcc(S, \mathcal{T}), \rw)$.
  \end{enumerate}
\end{definition}

\Cref{scc:dp} computes a minimum-cost connector.
For each SCC~$S \in \mathcal{S}_{G_\ell}$, we employ dynamic programming over vertices in~$S$ and subsets of~$\mathcal{S}_{G_\ell}$.
This gives us a significant speed-up compared to the na\"ive approach of branching into at worst~$c$ different neighbors on each vertex: 
from~$n^{\theta(c)}$~time to~$O(2^c \cdot c^2 \cdot n)$~time.

\begin{algorithm}[t]
\caption{A dynamic-programming procedure for the second phase}
\label{scc:dp}
\SetKwFunction{FDP}{DP}
\SetKwProg{Fn}{Function}{}{}
\SetKw{null}{null}

\Fn{\FDP{$\rsg, \rw$}}{
  \tcp{$S_v \in \mathcal{S}_{G_\ell}$ denotes the SCC to which~$v \in V(G_\ell)$ belongs}
  \tcp{$\mathcal{T}_B \subseteq \mathcal{S}_{G_\ell}$ denotes~$\{ S_u \mid \exists v.\; vu \in B \}$ for any~$B \subseteq A(G)$}
  \ForEach{$S = \{ v_1, \dots, v_{n_S} \} \in \mathcal{S}_{G_\ell}$}{
    \tcp{Initialization phase}
    $B_0 \leftarrow \emptyset$, $D_0(\emptyset) \leftarrow \emptyset$ \label{line:initstart}\;
    \ForEach{$i \in [n_S]$}{ \label{initstart}
      $B_i \leftarrow B_{i - 1} \cup \{ v_i u \in A(\rsg) \mid u \not\in S, S_u \not\in \mathcal{T}_{B_{i - 1}} \}$\label{biupdate}\;
      \lForEach{$\mathcal{T} \subseteq \mathcal{T}_{B_i}$}{$D_i(\mathcal{T}) \leftarrow \{ vu' \in B_i \mid S_{u'} \in \mathcal{T} \}$\label{initassign}}
    } \label{initend}
    \tcp{Update phase}
    \ForEach{$i \in \left[n_S\right]$}{ \label{veriter}
      \ForEach{$u \in \{ u' \mid u' \not\in S \wedge v_i u' \in A(\rsg) \}$}{
        $B_{i, u} \leftarrow \{ v_i u' \in A(\rsg) \mid u' \not\in S, \rw(v_i u') \le \rw(v_i u) \}$\label{niter}\;
        \ForEach{$\mathcal{T} \subseteq \sdom(S)$}{ \label{ssiter}
          \lIf{$\mathcal{T} \not\subseteq \mathcal{T}_{B_{i-1}} \cup \mathcal{T}_{B_{i, u}}$}{\Continue}
          \lIf{$\mathcal{T} \subseteq \mathcal{T}_{B_{i-1}}$}{$D_i(\mathcal{T}) \leftarrow D_{i - 1}(\mathcal{T})$\label{scc:dassign1}}
          \If{$\cost(S, D_{i-1}(\mathcal{T} \setminus \mathcal{T}_{B_{i, u}}), \rw) + \rw(v_i u) < \cost(S, D_{i}(\mathcal{T}), \rw)$}{
            $D_i(\mathcal{T}) \leftarrow D_{i-1}(\mathcal{T} \setminus \mathcal{T}_{B_{i, u}}) \cup \{ v_i u' \in B_{i, u} \mid S_{u'} \in \mathcal{T} \}$\label{scc:dassign2}
          }\label{niterend}
        }
      }
    }
    $\mcc(S, \, \cdot \,) \leftarrow D_{n_S}$ \;
  }
  \Return $\mcc$
}
\end{algorithm}

\begin{lemma}
  \label{lemma:dp}
  Given a relevant subgraph~$\rsg$ and a relevant weight function~$\rw$, \Cref{scc:dp} computes a minimum-cost connector~$\mcc$ in~$O(2^c \cdot c^2 \cdot n)$~time.
\end{lemma}
{
\begin{proof}
  We fix some~$S = \{ v_1, \dots, v_{n_S} \}$.
  We use~$S_v$ to denote the SCC to which the vertex~$v \in V(\rsg)$ belongs and~$\mathcal{T}_B = \{ S_u \mid \exists w.\; wu \in B \}$ to denote the SCCs containing a terminal vertex of at least one arc in~$B \subseteq A(\rsg)$.
  Note that in the subsequent proof the arcs in~$B$ will all have their initial vertex in~$S$.
  Moreover, the arc set~$B_i$ contains only arcs having a initial vertex in~$\{v_1, \ldots, v_i\} \subseteq S$ and we have~$\mathcal{T}_{B_i} = \{ S_u \mid \exists j \le i.\; v_ju \in A(\rsg)\}$.
  (See \Cref{biupdate} of \cref{scc:dp} for the computation of the arc set~$B_i$.)

  In order to show that the output of \Cref{scc:dp} satisfies the property of \Cref{def:mcc}~\ref{def:mcc1}, we prove the following stronger claim.
  \begin{claim}
    For any~$i \in \{ 0, \ldots, n_S \}$ and any set~$\mathcal{T} \subseteq \mathcal{T}_{B_i}$ of SCC, the arc set~$D_i(\mathcal{T})$ contains exactly one arc that starts in $\{v_1,\ldots,v_i\}$ and ends inside~$S'$ for each~$S' \in \mathcal{T}$ and contains no other arc.
  \end{claim}
  \begin{proof}[Proof of claim]
  We first show by induction that the claim holds after the initialization phase (\Crefrange{line:initstart}{initend}). 
  It holds for the base case~$i = 0$ because we have~$D_0(\emptyset) = \emptyset$.
  When~$i \ge 1$, since~$\rsg$ is a relevant subgraph (see \cref{def:rel}), it follows that the set~$\{ v_i u \in A(\rsg) \mid u \not\in S, S_u \not\in \mathcal{T}_{B_{i - 1}} \}$ on \Cref{biupdate} contains exactly one arc whose terminal vertex lies in~$S'$ for any~$S' \in \mathcal{T}_{B_i} \setminus \mathcal{T}_{B_{i-1}}$. %
  Since the claim holds for~$i - 1$ by induction hypothesis, $B_i$ contains exactly one arc that ends inside~$S'$ for each~$S' \in \mathcal{T}_{B_i}$ and no other arc.
  In \Cref{initassign}, the algorithm computes from this set~$B_i$ of arcs all arcs that end in~$\mathcal{T}$ and assigns this to~$D_i(\mathcal{T})$.
  Hence~$D_i(\mathcal{T})$ contains exactly one arc for each~$S' \in \mathcal{T}$ and thus the the claim holds for~$i$ as well.

  We now show---again by induction---that the claim holds after the update phase.
  Again the claim clearly holds for the base case~$i = 0$.
  When~$i \ge 1$, note that each iteration step (\Crefrange{niter}{niterend}) corresponds to the case in which the cost for~$v_i$ is exactly~$w(v_i u)$. 
  The algorithm finds the set~$B_{i, u}$ of arcs that~$v_i$ can cover with cost~$w(v_i u)$.
  We verify that the claimed property is maintained after the assignment on \Cref{scc:dassign1} and \Cref{scc:dassign2}.
  The assignment on \Cref{scc:dassign1} is clearly correct by induction hypothesis.
  Observe that~$B_{i, u}$ contains exactly one arc whose terminal vertex lies in~$S'$ for any~$S' \in \mathcal{T}_{B_{i, u}}$ since~$\rsg$ is a relevant subgraph.
  We can assume from the induction hypothesis that~$D_{i-1}(\mathcal{T} \setminus \mathcal{T}_{B_{i, u}})$ contains exactly one arc that ends inside~$S''$ for each~$S'' \in \mathcal{T} \setminus \mathcal{T}_{B_{i, u}}$ and no further arc.
  Thus, the assignment on \Cref{scc:dassign2} ensures that the claim is correct.
  \end{proof}
  
  For the second property, we prove by induction over~$i  \in \{ 0, \ldots, n_S \}$ that for any~$\mathcal{T} \in \sdom(S)$, the cost of~$S$ associated with~$D_i(\mathcal{T})$ is minimized when~$S$ is restricted to the vertices~$\{ v_1, \dots, v_i \}$.
  It holds in the base case~$i = 0$ because we have~$D_0(\emptyset) = \emptyset$.
  When~$i \ge 1$, we assume from the induction hypothesis that values~$\cost(S, D_{i-1}(\mathcal{T}), \rw)$ are minimum for any~$\mathcal{T} \subseteq \sdom(S)$.
  Assume towards a contradiction that there exists an arc set~$B'$ that satisfies~$\cost(S, D_{i}(\mathcal{T}_{B'}), \rw) > \cost(S, B', \rw)$.
  Let~$B_{i}' \subseteq B'$ denote the set of arcs that have~$v_i$ as their initial vertex.
  We distinguish two cases depending on whether~$B_{i}'$ is empty.

  Case 1. If~$B_{i}' = \emptyset$, then~$B'$ consists of arcs whose initial vertices are in~$\{ v_1 \dots, v_{i - 1} \}$.
  Then, we have 
  \[
    \cost(S, B', \rw) \ge \cost(S, D_{i - 1}(\mathcal{T}_{B'}), \rw) \ge \cost(S, D_i(\mathcal{T}_{B'}), \rw).
  \]
  Here the first inequality follows from the induction hypothesis, and the second inequality follows from the assignment on \Cref{scc:dassign1} because~$\mathcal{T}_{B'} \subseteq \mathcal{T}_{B_{i-1}}$.

  Case 2. If~$B_{i}' \ne \emptyset$, then the cost of~$S$ associated with~$B'$ is
  \begin{align*}
    \cost(S, B', \rw)
    &\ge \cost(S, \mathcal{D}_{i - 1}(\mathcal{T}_{B'} \setminus \mathcal{T}_{B_i'}), \rw) + \max_{v_i u \in B_{i}'} \rw(v_i u) \\ 
    &\ge \cost(S, D_i(\mathcal{T}_{B'}), \rw).
  \end{align*}
  Here the first inequality follows from the induction hypothesis, and the second inequality follows from the assignment on \Cref{scc:dassign2}.
  In both cases, we have~$\cost(S, B', \rw) \ge \cost(S, D_i(\mathcal{T}_{B'}), \rw)$, and thus a contradiction is reached.

  It remains to analyze the running time.
  In the initialization phase, we iterate over all vertices (\Cref{initstart}) and all subsets of~$\mathcal{T}_{B_i}$ (\Cref{initassign}).
  Note that~$|\mathcal{T}_{B_i}| \le |\mathcal{S}_{G_\ell}| = c$ as shown in the above claim and hence there are at most~$2^c$ subsets.
  Each iteration takes~$O(c)$ time because~$B_i$ contains at most~$c$ arcs for any~$i \in [n_S]$.
  In the update phase, we iterate over all vertices (\Cref{veriter}), at most~$c$ neighbors (\Cref{niter}), and all subsets of~$\sdom(S)$ (\Cref{ssiter}).
  Each iteration takes~$O(c)$ time because~$D_i(\mathcal{T})$ contains at most~$c$ arcs for any~$i \in [n_S]$ and any~$\mathcal{T} \subseteq \sdom(S)$.
  Thus, the overall running time is~$O(2^c \cdot c^2 \cdot n)$. 
\end{proof}
}

\paragraph{Phase 3.} We finally present the search tree algorithm for \textsc{MinPAC} in \Cref{alg:minpac}. 
The algorithm ``guesses'' the connections between SCCs of~$G_\ell$ to obtain an optimal solution.
To this end, we first try all possible numbers of outgoing arcs from each SCC. %
The array~$\mathcal{C}$ contains after \cref{mainend} for each SCC~$S_i$ in the~$i$\textsuperscript{th} entry all SCCs that~$S_i$ has an arc to in the solution.

\begin{algorithm}[t]
  \caption{An exhaustive search algorithm for \textsc{MinPAC}}
  \label{alg:minpac}

  \SetKwFunction{FSolveMinPAC}{Search}

  \Fn{\FSolveMinPAC{$\rsg, \rw, \mcc$}}{
    $OptCost \leftarrow \infty$, $\mathcal{C} \leftarrow (\mathcal{S}_{G_\ell}(S_1), \dots, \mathcal{S}_{G_\ell}(S_c))$\;
    \For{$k \leftarrow c$ \KwTo $2c-2$}{\label{eguess1}
      \ForEach{$k_1, \dots, k_c \in \N$ \emph{such that} $k_1, \dots, k_c \ge 1$ \emph{and} $\sum^c_{i=1} k_i = k$}{
        \ForEach{$\mathcal{T}_{S_1}, \dots, \mathcal{T}_{S_c}$ \emph{such that} $\mathcal{T}_{S_i} \subseteq \sdom(S_i)$ \emph{and} $|\mathcal{T}_{S_i}| = k_i$ \emph{for any} $i \in [c]$}{
          $H^{\mathrm{aux}} \leftarrow (\{ v_1, \dots, v_c \}, \{ v_i v_j \mid S_j \in \mathcal{T}_{S_i} \})$\;\label{phase3:iterstart}
          \lIf{$|\mathcal{S}_{H^{\mathrm{aux}}}| > 1$}{\Continue}\label{scc:count}
          $Cost \leftarrow 0$\;
          \lForEach{$S \in \mathcal{S}_{G_\ell}$}{
            $Cost \leftarrow Cost + \cost(S, \mcc(S, \mathcal{T}_S), \rw)$%
          }
          \tcp{We assume that $\cost(S, \mcc(S, \mathcal{T}_S), \rw)$ is computed for any $S \in \mathcal{S}_{G_\ell}, \mathcal{T} \subseteq \sdom(S)$ in \Cref{scc:dp}}
          \lIf{$Cost < OptCost$}{
            $OptCost \leftarrow Cost, \mathcal{C} \leftarrow ( \mathcal{T}_{S_1}, \dots, \mathcal{T}_{S_c} )$\label{phase3:iterend}
          }
        }
      }
    }\label{mainend}
    \Return{$(V(G), A(G_\ell) \cup \bigcup^c_{i=1} \mcc(S_i, \mathcal{C}[i]))$}
  }
\end{algorithm}

\begin{lemma}
  \label{fpt:parameterc}
  Given a relevant subgraph~$\rsg$, a relevant weight function~$\rw$, and a minimum-cost connector~$\mcc \colon \mathcal{S}_{G_\ell} \times 2^{\mathcal{S}_{G_\ell}} \to 2^{A(\rsg)}$, \Cref{alg:minpac} computes an optimal solution of~$(\rsg, \rw)$ in~$O(n + m + 4^c \cdot c^{2c - 3/2})$ time. 
\end{lemma}

{
\begin{proof}
  We first show the correctness of \Cref{alg:minpac}.
  In each iteration step (\Crefrange{phase3:iterstart}{phase3:iterend}), we 
  construct an auxiliary graph~$H^{\mathrm{aux}}$, which is basically a graph obtained from~$(V(G), \bigcup^c_{i=1} \mcc(S_i, \mathcal{T}_{S_i}))$ by contracting each SCC of~$G_\ell$ into a single vertex.
  Since our algorithm performs an exhaustive search, it finds a graph of cost~$\opt((\rsg, \rw))$.

  We now analyze the running time of the algorithm.
  For each~$c \le k \le 2c - 2$, the number of sets of integers~$\{ k_1, \dots, k_c \}$ that satisfy~$k_i \ge 1$ for all~$i \in [c]$ and~$\sum^c_{i=1} k_i = k$ is
  \begin{align*}
    \binom{k - 1}{c - 1} & \le \binom{2c - 3}{c - 1}
    = \frac{(2c - 3)!}{(c - 1)! \, (c - 2)!} \\
    &\in O\left( \frac{\sqrt{2c - 3}}{\sqrt{c - 1}\sqrt{c - 2}} \cdot \frac{(2c - 3)^{2c - 3}}{(c - 1)^{c - 1} (c - 2)^{c - 2}} \right)
    = O \left(\frac{4^c}{\sqrt{c}} \right),
  \end{align*}
  where the membership is due to Stirling's approximation. 
  For each fixed set of integers~$\{k_1, \dots, k_c\}$, the number of sets~$\{ \mathcal{T}_{S_1}, \dots, \mathcal{T}_{S_c} \}$ the algorithm generates is
  \begin{align*}
    \prod^c_{i = 1} \binom{c}{k_i} \le \prod^c_{i =1} c^{k_i} = c^{\sum^c_{i=1} k_i} = c^k.
  \end{align*}
  So the total number of iterations (\Crefrange{phase3:iterstart}{phase3:iterend}) is
  \begin{align*}
    O\left(\frac{4^c}{\sqrt{c}} \cdot \sum^{2c - 2}_{k = c} c^k\right)
    = O(4^c \cdot c^{2c - 5/2}).
  \end{align*}
  We claim that each iteration step runs in~$O(c)$ time.
  Computing the SCCs of~$H^{\mathrm{aux}}$ takes~$O(c)$ time because~$H^{\mathrm{aux}}$ contains~$c$ vertices and at most $2c - 2$~arcs and the algorithm spends~$O(c)$ time to compute the cost of~$H$ and update~$OptCost$ and connection~$\mathcal{C}$.
  Constructing the output graph takes~$O(n + m)$ time.
  Thus, the overall running time is~$O(n + m + 4^c \cdot c^{2c - 3/2})$. 
\end{proof}
}

Combining \Cref{scc:redalgo,scc:dp,alg:minpac} we arrive at our main theorem of this section.

\begin{theorem}
    \label{thm:SCC-Algo}
    \textsc{MinPAC} can be solved in~$O(c^2 \cdot 2^c \cdot n + m + 4^c \cdot c^{2c - 3 / 2})$~time.
\end{theorem}
{
\begin{proof}
  Let~$(G, w)$ be an instance of \textsc{MinPAC}.
  We run \Cref{scc:redalgo,scc:dp,alg:minpac} sequentially to obtain an optimal solution~$H$ of~$(\rsg, \rw)$.
  Since~$\rsg$ is a relevant subgraph of~$G$ it follows from \Cref{scc:arcred} that the graph~$H$ is also an optimal solution of~$(G, w)$.

  The overall running time is then
  \begin{align*}
    O(cn + m) + O(2^c \cdot c^2 \cdot n) + O(n + m + 4^c \cdot c^{2c - 3/2}) \\
    = O(c^2 \cdot 2^c \cdot n + m + 4^c \cdot c^{2c - 3 / 2}).
  \end{align*}
  This concludes the proof of the theorem.
\end{proof}
}

\section{Parameterization by the number of power levels} \label{sec:vcn}

It is fair to assume that the nodes cannot transmit signals with arbitrary power levels due to practical limitations \cite{Carmi2007}. 
In fact, many researchers have studied approximation algorithms for the \textsc{MinPAC} problem when only two power levels are available \cite{Rong2004,Carmi2007, Calinescu2013approximate,Calinescu2016}.
In this section, we consider the case~${w \colon A(G) \to Q}$, where the set of integers~$Q = \{ p_1, \dots, p_q \}$ represents available power levels.
The parameter~$q$---``the number of numbers''---has been advocated by \citet{FGR12}.
The problem remains NP-hard even when~$q = 2$ \cite{chen1989strongly}, as also can be seen in our hardness result (\Cref{theo:hardness}).
Thus, fixed-parameter tractability is unlikely with this parameter alone. 
However, using an additional parameter may alleviate this problem.
We consider the \emph{vertex cover number}, as many problems are known to become tractable when this parameter is bounded.
Here we define the vertex cover number for a directed graph as the vertex cover number of the underlying undirected graph.
Recall that the vertex cover number for an undirected graph is the minimum number of vertices that have to be removed to make it edgeless.
Computing a minimum-cardinality vertex cover is NP-hard 
but 
any maximal matching (which can be found in linear time) gives a factor\nobreakdash-2 approximation.
We present a partial kernelization (unbounded weights) with respect to~$q + x$, where~$x$ is the size of a given vertex cover.
Afterwards, we strengthen this result to a proper polynomial kernel (with a worse but still polynomial running time).

\begin{theorem} \label{thm:VCKern}
  Let~$I = (G, w)$ be a \textsc{MinPAC}-instance where~$w \colon A(G) \to Q$ and~$Q \in \N^q$.
  Given~$I$ and a vertex cover~$X$ for~$G$ of size~$x$, one can compute an instance~$I'$ of \textsc{MinPAC} with at most~${(q + 1)^{2 x} + x}$ vertices and a value~$d \in \N$ such that~$\opt(I) = \opt(I') + d$ in~$O(\min(\{ x, \log_q n \}) \cdot n + m)$ time.
\end{theorem}
{
In order to prove \Cref{thm:VCKern}, we first observe that there are some conditions under which a vertex can be included or removed without losing the strong connectivity.
Notably, we use \Cref{obs:remv} to remove ``twin'' vertices.

\begin{observation} \label{obs:remv}
  Let~$G$ be a strongly connected graph with~$u \in V(G)$.
  If there exists a vertex~$u' \in V(G)$ with~$N^+_G(u) \subseteq N^+_G(u')$ and~$N^-_G(u) \subseteq N^-_G(u')$, the graph~$G[V(G) \setminus \{ u \}]$ is strongly connected.
\end{observation}

\begin{observation} \label{obs:addv}
  Let~$G$ be a strongly connected graph.
  For any vertices~$v, v' \in V(G)$, $u \notin V(G)$, the graph~$G + uv + v'u$ is strongly connected. 
\end{observation}

We define ``types'' for vertices outside the vertex cover according to the weights of their incident arcs.
This helps us to reduce the number of vertices using the observations above. 
Recall that we write~$w(vu) = \infty$ if~$vu \not\in A$. 

\begin{definition}
  Let~$(G, w)$ be an instance of \textsc{MinPAC} with~$w \colon A(G) \to Q$ and~$Q = \{ p_1, \dots, p_q \} \in \N^q$.
  Let~$X = \{ v_1, \dots, v_x \}$ be a vertex cover of~$G$.
  The \emph{vertex cover partition} is the partition~$\mathcal{P}$ of vertices in~$V(G) \setminus X$ into sets
  \begin{align*}
    \mathcal{P}_{r_1, \dots, r_{2x}} = \{ u \in V(G) \setminus X \mid \forall i \in [x].\;w(u v_i) = p_{r_i} \wedge w(v_i u) = p_{r_{i+x}}  \},
  \end{align*}
  for each~$r_1, \dots, r_{2x} \in [q + 1]$.
  Here we set~$p_{q+1} = \infty$.
\end{definition}

We initialize~$d$ with 0.
In our reduction rule, we remove vertices such that, after the reduction is completed,  there is at most one vertex in each set of the vertex cover partition.

\begin{reductionrule} \label{red:vcr}
  Let~$\mathcal{P}_{r_1, \dots, r_{2x}}$ be a set of the vertex cover partition with $|\mathcal{P}_{r_1, \dots, r_{2x}}| > 1$.
  Delete an arbitrary vertex~$u \in \mathcal{P}_{r_1, \dots, r_{2x}}$ and increase~$d$ by~$\min_{i \in [x]} p_{r_i}$. 
\end{reductionrule}

Note that since the input graph is strongly connected, the increase in $d$ in \Cref{red:vcr} is at most $\max_{i \in [q]} p_i < \infty$.

\begin{lemma}
  \Cref{red:vcr} is correct.
\end{lemma}
\begin{proof}
  Let~$I = (G, w)$ be an instance of \textsc{MinPAC} and let~$I'= (G', w')$ be the instance obtained by deleting vertex~$u$ in a set~$\mathcal{P}_{r_1, \dots, r_{2x}}$ of the vertex cover partition for~$G$, as specified in \Cref{red:vcr}.
  We show that~$\opt(I) = \opt(I') +  \min_{i \in [x]} p_{r_i}$.

  Let~$H$ be an optimal solution of~$I$.
  Since~$|\mathcal{P}_{r_1, \dots, r_{2x}}| > 1$, there exists a vertex~$u' \in \mathcal{P}_{r_1, \dots, r_{2x}} \setminus \{ u \}$.
  We can assume without loss of generality that~$\max_{uv \in A(H)} w(uv) \le \max_{(u'v) \in A(H)} w(u'v)$ holds: if it does not hold, we can exchange the role of~$u$ and~$u'$ in~$H$ without changing the cost of the solution, that is, we can update~$H$ to~$H'$ with
  \begin{align*}
    A(H') := {}& \{vv' \mid vv' \in A(H) \wedge \{v,v'\} \cap \{u,u'\} = \emptyset \} \cup \{uv \mid u'v \in A(H)\}  \\
             & \cup \{vu \mid vu' \in A(H)\} \cup \{u'v \mid uv \in A(H)\} \cup \{vu' \mid vu \in A(H)\}.   
  \end{align*}
  Then, we can assume that~$N^+_H(u) \subseteq N^+_H(u')$ and~$N^-_H(u)\subseteq N^-_H(u')$ hold (otherwise we can add the missing arcs to~$H$ without additional cost). 
  Then, it follows from \Cref{obs:remv} that~$G[V(G')]$ is a solution of~$I'$.
  Its cost is at most~$\opt(I) - \min_{i \in [x]} p_{r_i}$ because~$u$ pays at least~$\min_{i \in [x]} p_{r_i}$ in~$H$.
  This shows that~$\opt(I) \ge \opt(I') + \min_{i \in [x]} p_{r_i}$.
  For the other direction, suppose that~$H'$ is an optimal solution of~$I'$.
  Let~$u' \in \mathcal{P}_{r_1, \dots, r_{2x}} \setminus \{ u \}$ be a vertex and let~$v$ be a vertex with~$w(u'v) = \min_{i \in [x]} p_{r_i}$.
  Since~$H'$ is strongly connected, there exists a vertex~$v' \in X$ with~$v' u' \in A(H')$.
  Due to \Cref{obs:addv}, $H' + uv + v'u$ is strongly connected.
  Observe that the cost for~$u$ is~$\min_{i \in [x]} p_{r_i}$ and the cost for~$v'$ remains unchanged because~$v'$ pays at least~$w(v'u') = w(v'u)$ in~$H'$.
  Hence, $\opt(I) = \opt(I') + \min_{i \in [x]} p_{r_i}$.
\end{proof}

We have shown that \Cref{red:vcr} is correct.
It remains to show that it can be applied in~$O(xn + m)$ or $O(n \log_q n + m)$ time to complete the proof of \Cref{thm:VCKern}.

\begin{proof}[Proof of \Cref{thm:VCKern}]
  We present a procedure that transforms an instance of \textsc{MinPAC}~$I = (G, w)$ into another instance~$I' = (G', w')$ with at most~$(q + 1)^{2x} + x$ vertices in~$O(\min(\{ x, \log_q n \}) \cdot n + m)$ time.
  In our transformation, we distinguish two cases depending on the input size. 

  Case 1.
  If~$n \le (q + 1)^{2x} + x$, then we return~$I$ as the output of the transformation with~$d = 0$. 

  Case 2.
  If~$n > (q + 1)^{2x} + x$, then we apply \Cref{red:vcr} exhaustively.
  Let~$\mathcal{P}$ be the vertex cover partition.
  There are at most~$(q + 1)^{2x}$ sets of~$\mathcal{P}$ and each set yields at most one vertex in~$G'$.
  Thus, the reduced instance has at most~$(q + 1)^{2x} + x$ vertices.
  We show that the transformation can be performed in~$O(xn + m)$ time.
  We first build a~$2x$-dimensional table~$D$, where for each dimension there are~$q + 1$ values.
  All the~$(q+1)^{2x}$ entries of~$D$ are initialized as false.
  (This can be done in~$O(n)$ time, since $n > (q + 1)^{2x} + x$.)
  The entry~$D[r_1,\dots,r_{2x}]$ represents whether a vertex in the set~$\mathcal{P}_{r_1, \dots, r_{2x}}$ has been found.
  We iterate through all vertices in~$V(G) \setminus X$.
  For each vertex~$u \in V(G) \setminus X$, we set the corresponding entry in~$D$ to true if it is false, and we remove~$u$ and its incident arcs if it is true.
  Since accessing an entry in~$D$ takes~$O(x)$ time and removing~$u$ takes~$O(\deg_G(u))$ time, the transformation overall takes~$O(xn + m)$~time. 
  Note that $n > (q + 1)^{2x} + x$ yields that $2x < \log_q n$ and hence the reduction can also be done in~$O(n \log_q n + m)$ time.
\end{proof}
}

Notice that \Cref{thm:VCKern} does not show a kernel for the parameter combination vertex cover~$x$ plus number of numbers~$q$.
In order to obtain a kernel, we will next show how to shrink the weights.

\begin{theorem}
    \label{thm:weights}
    Let~$I=(G,w)$ be an instance of \textsc{MinPAC} where~$G$ contains~$n$ vertices and~$m$ edges.
    There is a polynomial-time algorithm that computes a new weight function~$\hat{w}$ such that~$||\hat{w}||_{\infty} < 2^{4m^3}(4nm+1)^{m (m+2)}$ and such that any optimal solution~$T=(V,F)$ of~$(G,w)$ is also an optimal solution for~$(G,\hat{w})$.
\end{theorem} 

\begin{proof}
    We use the notion of~$\alpha$-$\Z$-linearizability, which uses
    
    \[ \Z_r := \{-r, -r+1, \ldots, r-1, r\} \]
    
    and is defined as follows:
    
    \begin{definition}[{\cite{BBFNN}}]
        A function~$f\colon L\times \mathds{Q}^d\to \mathds{Q}$ with~$L\subseteq \Sigma^*$ is \Zlin{$\alpha$}, 
        $\alpha\in \N$, 
        if for all~$\omega\in\mathds{Q}^d$ and for all~$x\in L$ it holds that 
        \begin{enumerate}
        \item there exists~$b_x\in\Z_\alpha^d$ 
            such that~$f(x,\omega)=b_x^\top \omega$ and\label{alphZlin:A}
        \item for all~$\omega'\in \{\omega'' \in \mathds{Q}^d \mid \forall \beta \in \Z_\alpha^d.\; \sign(\beta^T \omega) = \sign(\beta^T \omega'')\}$ it holds that
                $f(x,\omega)=b_x^\top \omega$ if and only if~$f(x,\omega')=b_x^\top \omega'$.\label{alphZlin:B}
        \end{enumerate}
    \end{definition}

    To this end, observe that we can rewrite the goal function to fit their notion as follows.
    Let~$F_v:=\{vu\in F\mid u \in N_G^+(v)\}$ and~$\mathcal{F}:=\{F_v\mid v\in V\}$.
    Then 
    \[ \cost(V, F, w) = \sum_{F_v\in \mathcal{F}} g(F_v,w),\quad\text{ with } g(F,w)=\max_{e\in F} w(e).\]
    Clearly, 
    with~$A=\{e_1,\ldots,e_m\}$ 
    the function~$f \colon E\times\Z^m\to \Z$, $f(e_i,\omega)\mapsto \omega_i := w(e_i)$ 
    is \Zlin{1}:
    On the one hand,
    we have that~$f(e_i,\omega)=\vec{e}_i^\top \omega$
    (where~$\vec{e}_i$ denotes the unit vector with the~$i$\textsuperscript{th} entry being one).
    On the other hand,
    for all~$\omega'\in \Z^m$
    it holds true that~$f(e_i,\omega)=\vec{e}_i^\top \omega$ if and only if~$f(e_i,\omega')=\vec{e}_i^\top \omega'$.
    
    By Lemma~4.8 in \cite{BBFNN}, it follows that~$\cost(V, F, w)$ is~\Zlin{$2n$}.
    Finally, Theorem~4.7 in \cite{BBFNN} yields the desired weight function~$\hat{w}$.
    \qed
\end{proof}

Combining \Cref{thm:VCKern,thm:weights} gives us the desired kernel.

\begin{corollary}
    \label{cor:expKernelVCq}
    \textsc{MinPAC} admits an exponential-size kernel with respect to the combined parameter vertex cover plus number of numbers.
\end{corollary}

\section{Parameterization by feedback edge number}
\label{sec:fen}

In this section we describe a kernelization for \textsc{MinPAC} parameterized by the \emph{feedback edge number}.
The feedback edge number for an undirected graph is the minimum number of edges that have to be removed in order to make it a forest.
We define the feedback edge number for a directed graph~$G$ as the feedback edge number of its underlying undirected graph~$\uug{G}$.
Note that a minimum feedback edge set can be computed in linear time.
In \Cref{section:hardness}, we will show that the parameter \emph{feedback arc number}, which is the directed counterpart of the feedback edge number, does not allow the design of an FPT algorithm for \textsc{MinPAC} unless P = NP.

The feedback edge number measures how tree-like the input is.
From a theoretical perspective this is interesting to analyze because any instance~$(G, w)$ of \textsc{MinPAC} is easy to solve if~$\uug{G}$ is a tree.
In this case all edges of~$\uug{G}$ must correspond to arcs in both directions in~$G$ and the optimal solution is~$G$ itself.
The parameter is also motivated by real world applications in which the feedback edge number is small; for instance, sensor networks along waterways (including canals) are expected to have a small number of feedback edges.
In this section we first prove the following theorem which states that \textsc{MinPAC} admits a partial kernel with respect to feedback edge number.
Afterwards, we strengthen this result to a proper polynomial kernel (with a worse but still polynomial running time).

\begin{theorem} \label{theorem:kernelization}
  In linear time, one can transform any instance~$I=(G,w)$ of \textsc{MinPAC} with feedback edge number g into an instance~$I'=(G',w')$ and compute a value~$d \in \N$ such that~$G'$ has at most~$20g-20$ vertices, $42g-42$ arcs, and~$\opt(I)=\opt(I') + d$.
\end{theorem}

\begin{corollary}
    \textsc{MinPAC} can be solved in~$O(2^{O(g)} + n + m)$ time.
\end{corollary}

We will present a set of data reduction rules which shrink any instance of \textsc{MinPAC} to an essentially equivalent instance whose size is bounded as specified in \Cref{theorem:kernelization}.
We simultaneously compute the value~$d$, which we initialize with 0.

Our first reduction rule reduces the weights of arcs outgoing from a vertex by the weight of its cheapest outgoing arc.
This ensures that each vertex has at least one outgoing arc of weight zero.

\begin{reductionrule} \label{reduction:costs}
  Let~$v$ be a vertex with~$\delta_v := \min_{vu \in A(G)} w(vu) > 0$.
  Update the weights and~$d$ as follows:
  \begin{enumerate}[align=left,label={(\roman*)}]
    \setlength{\itemindent}{.5em}
    \setlength{\labelwidth}{1.5em}
    \item $w(vu) = w(vu) - \delta_v$ for each~$vu \in A(G)$.
    \item $d := d + \delta_v$.
  \end{enumerate}
\end{reductionrule}

{
\begin{lemma}
  \label{lemma:cost}
  \Cref{reduction:costs} is correct.
\end{lemma}

\begin{proof}
  Let~$I=(G,w)$ be an instance of \textsc{MinPAC} and let~$v \in V(G)$ be a vertex with~$\delta_v = \min_{vu \in A(G)} w(vu) > 0$.
  Let~$I'$ be a instance with reduced weights using \Cref{reduction:costs}.
  We show that~$\opt(I) = \opt(I') + \delta_v$.

  Let~$H$ be an optimal solution of~$I$.
  Then, $H$ is also a solution of~$I'$, where the cost for~$v$ is decreased by~$\delta_v$ and the cost for every other vertex remains identical.
  Thus,~$\opt(I')$ is at most~$\opt(I) - \delta_v$ and we obtain~$\opt(I) \geq \opt(I') + \delta_v$.
  For the other direction, let~$H$ be an optimal solution for~$I'$.
  Then~$H$ is also a solution for~$I$, where the cost for~$v$ is increased by~$\delta_v$ and the cost for every other vertex remains identical.
  Thus,~$\opt(I)$ is at most~$\opt(I') + \delta_v$ and we obtain~$\opt(I) = \opt(I') + \delta_v$.
\end{proof}
}
Our next reduction rule discards all degree-one vertices.

\begin{reductionrule} \label{reduction:degree1}
  Let~$v$ be a vertex with~$\deg_G(v) = 1$ and let~$u$ be its neighbor.
  Update~$(G, w)$ and~$d$ as follows:
  \begin{enumerate}[align=left,label={(\roman*)}]
    \setlength{\itemindent}{.5em}
    \setlength{\labelwidth}{1.5em}
    \item $G := G[V(G) \setminus \{ v \}]$.
    \item $w(u v') := \max \{ 0, w(uv') - w(uv) \}$ for each~$uv' \in A(G) \setminus \{ uv \}$.
    \item $d := d + w(vu) + w(uv)$.
  \end{enumerate}
\end{reductionrule} 

{
\begin{lemma}
  \label{lemma:degree1correct}
  \Cref{reduction:degree1} is correct.
\end{lemma}

\begin{proof}
  Let~$I = (G, w)$ be an instance of \textsc{MinPAC} with an optimal solution~$H$.
  Let~$v \in V(G)$ be a vertex with~$\mathrm{deg}_G(v) = 1$ and~$u$ be its neighbor.
  (Since~$G$ is strongly connected, we have~$uv \in A(G)$ and~$vu \in A(G)$.)
  Let~$I' = (G', w')$ be the instance in which~$v$ is removed according to \Cref{reduction:degree1}.
  Then,~$H[V(G')]$ is a solution of~$I'$.
  Since the cost for~$u$ decreases by~$\max_{uv' \in A(H)} w(uv') - \max_{uv' \in A(H) \setminus \{ uv \}} w'(uv') = w(uv)$ and the costs for other vertices remain unchanged, we have~$\opt(I) \ge \opt(I') + w(vu) + w(uv)$.
  For the other direction, let~$H'$ be an optimal solution of~$I'$.
  Then, $H' + vu + uv$ is a solution of~$I$.
  The cost for~$v$ is~$w(vu)$ and the cost for~$u$ is~$\max_{uv' \in A(H') \cup \{ uv \}} w(uv') = w(uv) + \max_{uv' \in A(H')} w'(uv')$, while the costs for other vertices remain the same.
  Thus, we obtain~$\opt(I) \le \opt(I') + w(vu) + w(uv)$.
\end{proof}
}

\begin{lemma}
  \Cref{reduction:costs,reduction:degree1} can be exhaustively applied in linear time.
\end{lemma}
\begin{proof}
  For each vertex~$v \in V(G)$, set~$\ell(v) := \min_{vu \in A(G)} w(vu)$.
  Let~$L$ be a list of degree-1 vertices.
  We apply the following procedure as long as~$L$ is nonempty.
  Let~$v$ be the vertex taken from~$L$ and let~$u$ be its neighbor.
  Remove~$v$ and its incident arcs from~$G$, set~$\ell(u) := \max \{ \ell(u), w(uv) \}$, and update~$d := d + \max \{ w(vu), \ell(v) \}$.
  If the degree of~$u$ becomes~$1$ after deleting~$v$, then add~$u$ to~$L$.
  Once~$L$ is empty, update the weight of each remaining arc~$w(vu) := \max \{ 0, w(vu) - \ell(v) \}$.
  Finally, update~$d := d + \ell(v)$ for each remaining vertex~$v$.
  It is easy to see that the algorithm runs in linear time.
\end{proof}

Henceforth, we can assume that \Cref{reduction:costs,reduction:degree1} are exhaustively applied.
Thus, the underlying undirected graph~$\uug{G}$ will have no degree-one vertices. 
It remains to bound the number of vertices that have degree two in~$\uug{G}$.
Once this is achieved, we can use standard arguments to upper-bound the size of the instance~\cite{BBNN17}.

The rough idea to bound the number of degree-two vertices is as follows:
In order to upper-bound the number of degree-two vertices in~$\uug{G}$, we consider long paths in~$\uug{G}$.
A path~$P = (v_0, \dots, v_{h+1})$ in~$\uug{G}$ is a \emph{maximal induced path} of~$G$ if~$\deg_G(v_0) > 2$, $\deg_G(v_{h+1}) > 2$, and~$\deg_G(v_i) = 2$ for all~$i \in [h]$.
We call the vertices~$\{ v_i \mid i \in [h] \}$ the inner vertices of~$P$.
We will replace the inner vertices of each maximal induced path on at least seven vertices with a fixed gadget. %
The arc-weights in the gadget are chosen such that the four possible ways in which the outermost inner vertices are connected inside the path (see \Cref{fig:cases} for a visualization of the four cases) are preserved.

\begin{figure}[t]
  \tikzset{svertex/.style = {shape=circle,draw,inner sep = 2pt}}
  \tikzset{dotvertex/.style = {shape=circle,inner sep = 2pt}}
  \captionsetup[subfigure]{labelformat=empty}
  \centering
  \begin{subfigure}{.5\textwidth}
    \centering
    \begin{tikzpicture}[scale=0.4]
      \node[svertex] (v1) at (0, 0) {};
      \node at (0,1) {$v_1$};
      \node[svertex] (v2) at (2, 0) {};
      \node[svertex] (v3) at (4, 0) {};
      \node[dotvertex] (dots) at (6,0) {...};
      \node[svertex] (vh_3) at (8, 0) {};
      \node[svertex] (vh_2) at (10, 0) {};
      \node[svertex] (vh_1) at (12, 0) {};
      \node at (12, 1) {$v_{h}$};
      \path[->] (v1) edge [bend left] (v2);
      \path[->] (v2) edge [bend left] (v3);
      \path[->] (v3) edge [bend left] (dots);
      \path[->] (dots) edge [bend left] (vh_3);
      \path[->] (vh_3) edge [bend left] (vh_2);
      \path[->] (vh_2) edge [bend left] (vh_1);
      \path[->, densely dotted] (v2) edge [bend left] (v1);
      \path[->, densely dotted] (v3) edge [bend left] (v2);
      \path[->, densely dotted] (dots) edge [bend left] (v3);
      \path[->, densely dotted] (vh_3) edge [bend left] (dots);
      \path[->, densely dotted] (vh_2) edge [bend left] (vh_3);
      \path[->, densely dotted] (vh_1) edge [bend left] (vh_2);
     \end{tikzpicture}
     \caption{(R)}
  \end{subfigure}%
  \begin{subfigure}{.5\textwidth}
    \centering
    \begin{tikzpicture}[scale=0.4]
      \node[svertex] (v1) at (0, 0) {};
      \node at (0,1) {$v_1$};
      \node[svertex] (v2) at (2, 0) {};
      \node[svertex] (v3) at (4, 0) {};
      \node[dotvertex] (dots) at (6,0) {...};
      \node[svertex] (vh_3) at (8, 0) {};
      \node[svertex] (vh_2) at (10, 0) {};
      \node[svertex] (vh_1) at (12, 0) {};
      \node at (12, 1) {$v_{h}$};
      \path[->] (v2) edge [bend left] (v1);
      \path[->] (v3) edge [bend left] (v2);
      \path[->] (dots) edge [bend left] (v3);
      \path[->] (vh_3) edge [bend left] (dots);
      \path[->] (vh_2) edge [bend left] (vh_3);
      \path[->] (vh_1) edge [bend left] (vh_2);
      \path[->, densely dotted] (v1) edge [bend left] (v2);
      \path[->, densely dotted] (v2) edge [bend left] (v3);
      \path[->, densely dotted] (v3) edge [bend left] (dots);
      \path[->, densely dotted] (dots) edge [bend left] (vh_3);
      \path[->, densely dotted] (vh_3) edge [bend left] (vh_2);
      \path[->, densely dotted] (vh_2) edge [bend left] (vh_1);
     \end{tikzpicture}
     \caption{(L)}
  \end{subfigure}%
  \\
  \begin{subfigure}{.5\textwidth}
    \centering
    \begin{tikzpicture}[scale=0.4]
      \node[svertex] (v1) at (0, 0) {};
      \node at (0,1) {$v_1$};
      \node[svertex] (v2) at (2, 0) {};
      \node[svertex] (v3) at (4, 0) {};
      \node[dotvertex] (dots) at (6,0) {...};
      \node[svertex] (vh_3) at (8, 0) {};
      \node[svertex] (vh_2) at (10, 0) {};
      \node[svertex] (vh_1) at (12, 0) {};
      \node at (12, 1) {$v_{h}$};
      \path[->] (v1) edge [bend left] (v2);
      \path[->] (v2) edge [bend left] (v3);
      \path[->] (v3) edge [bend left] (dots);
      \path[->] (dots) edge [bend left] (vh_3);
      \path[->] (vh_3) edge [bend left] (vh_2);
      \path[->] (vh_2) edge [bend left] (vh_1);
      \path[->] (v2) edge [bend left] (v1);
      \path[->] (v3) edge [bend left] (v2);
      \path[->] (dots) edge [bend left] (v3);
      \path[->] (vh_3) edge [bend left] (dots);
      \path[->] (vh_2) edge [bend left] (vh_3);
      \path[->] (vh_1) edge [bend left] (vh_2);
     \end{tikzpicture}
     \caption{(B)}
  \end{subfigure}%
  \begin{subfigure}{.5\textwidth}
    \centering
    \begin{tikzpicture}[scale=0.4]
      \node[svertex] (v1) at (0, 0) {};
      \node[svertex] (v2) at (1.8, 0) {};
      \node[dotvertex] (v3) at (3.6, 0) {...};
      \node[svertex] (vk) at (5.4, 0) {};
      \node[svertex] (vk1) at (6.6,0) {};
      \node[dotvertex] (vh_3) at (8.4, 0) {...};
      \node[svertex] (vh_2) at (10.2, 0) {};
      \node[svertex] (vh_1) at (12, 0) {};

      \node at (0, 1) {$v_1$};
      \node at (5.4, 1) {$v_{k}$};
      \node at (6.8, 1) {$v_{k+1}$};
      \node at (12, 1) {$v_{h}$};

      \path[->] (v1) edge [bend left] (v2);
      \path[->] (v2) edge [bend left] (v1);
      \path[->] (v2) edge [bend left] (v3);
      \path[->] (v3) edge [bend left] (v2);
      \path[->] (v3) edge [bend left] (vk);
      \path[->] (vk) edge [bend left] (v3);
      \path[->] (vk1) edge [bend left] (vh_3);
      \path[->] (vh_3) edge [bend left] (vk1);
      \path[->] (vh_3) edge [bend left] (vh_2);
      \path[->] (vh_2) edge [bend left] (vh_3);
      \path[->] (vh_2) edge [bend left] (vh_1);
      \path[->] (vh_1) edge [bend left] (vh_2);
     \end{tikzpicture}
     \caption{(N)}
  \end{subfigure}%
  \caption{Visualization of the four cases for connectivity inside maximal induced paths (see \Cref{obs:cases}).
  For cases (R) and (L), at least one of the dotted arcs is not present.
  }
  \label{fig:cases}
\end{figure}
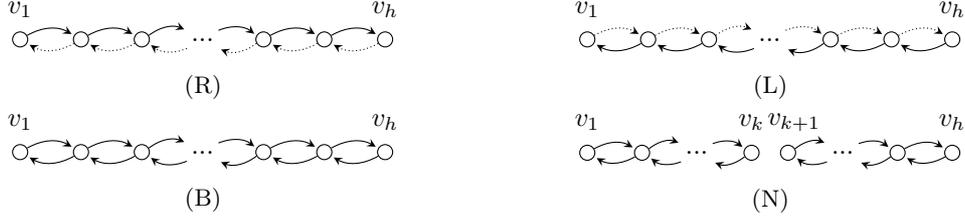

{

\begin{observation}
  \label{obs:cases}
  Let~$I=(G,w)$ be an instance of \textsc{MinPAC} with an optimal solution~$H$.
  Let~$P = (v_0, \dots, v_{h+1})$ be a maximal induced path of~$G$.
  Then, there are four cases in which~$v_1$ and~$v_h$ are connected inside~$P$ in~$H$ (\Cref{fig:cases}):
  \begin{description}
    \item[Case (R)] It holds for any~$i \in [h - 1]$ that~$v_i v_{i+1} \in A(H)$ and there exists~$k \in [h - 1]$ such that~$v_{k+1} v_k \not\in A(H)$.
    \item[Case (L)] It holds for any~$i \in [h - 1]$ that~$v_{i+1} v_i \in A(H)$ and there exists~$k \in [h - 1]$ such that~$v_k v_{k+1} \not \in A(H)$.
    \item[Case (B)] It holds for any~$i \in [h - 1]$ that~$v_i v_{i+1} \in A(H)$ and~$v_{i+1} v_i \in A(H)$.
    \item[Case (N)] There exists~$k \in [h - 1]$ such that~$v_k v_{k+1} \not\in A(H)$ and~$v_{k+1} v_k \not\in A(H)$, and thus it holds for any~$i \in \{ 1, \dots, k - 1, k + 1, \dots, h - 1 \}$ that~$v_i v_{i+1} \in A(H)$ and~$v_{i+1} v_i \in A(H)$.
  \end{description}
\end{observation}

\begin{proof}
  If~$v_i v_{i+1} \in A(H)$ holds for all~$i \in [h - 1]$ or~$v_{i+1} v_i \in A(H)$ holds for all~$i \in [h - 1]$, then we have one of the three cases (R), (L), or (B).
  Otherwise there exist~$k, k' \in [h - 1]$ such that~$v_k v_{k + 1} \not\in A(H)$ and~$v_{k'+1} v_{k'} \not\in A(H)$.
  We show that this corresponds to the case (N).
  To this end, we show that~$k = k'$.
  If~$k < k'$, then there is no outgoing arc from any vertices of~$S_1 = \{ v_{k+1}, \dots v_{k'} \}$ to~$V(G) \setminus S_1$.
  This is contradicting the assumption that~$H$ is a solution, and hence~$k \ge k'$.
  If~$k > k'$, then there is no incoming arc to any vertices of~$S_2 = \{ v_{k'+1}, \dots, v_k \}$ from~$V(G) \setminus S_2$.
  Again this is a contradiction, and hence we obtain~$k \le k'$.
\end{proof}

Before giving the gadget to replace the inner vertices of a maximal induced path, we define the cost for the inner vertices in cases (R), (L), and (N).

\begin{definition} \label{def:repGadgetWeights}
  Let~$(G, w)$ be an instance of \textsc{MinPAC} and let~$P = ( v_0, \dots, v_{h+1} )$ be a maximal induced path of~$G$.
  We define the cost for the connection inside~$P$ in the right direction, the left direction, and neither direction as follows:
  \begin{align*}
    C_{R} :=\sum_{i = 1}^{h - 1} w(v_{i} v_{i+1}),\;\;\;\; C_{L} := \sum_{i=1}^{h-1} w(v_{i+1}v_{i}),\\ C_{N} := \sum_{i = 1, \, i \ne k}^{h-1}{w(v_{i} v_{i+1}) + w(v_{i+1} v_{i})},
  \end{align*}
  where
  \[
    k := \mathop{\mathrm{argmax}}_{i \in [h - 1]} ( w(v_i v_{i + 1}) + w(v_{i + 1} v_i) ). 
  \]
\end{definition}

Note that~$C_R = \infty$ (or~$C_L = \infty$) if~$v_i v_{i+1} \not\in A(G)$ (or~$v_{i+1} v_i \not\in A(G)$) for some~$i$ (recall that~$w(vu) = \infty$ for~$vu \not\in A(G)$).
We finally present the gadget to replace the inner vertices of a maximal induced path.
The gadget is somewhat more involved than the gadget used in the symmetric version of \textsc{MinPAC} \cite{BBNN17} because it needs to encode the four cases seen in \Cref{obs:cases}.

\begin{definition} \label{def:repGadget}
  Let~$P=( v_0, \dots, v_{h+1} )$ be a maximal induced path.
  The path-gadget for~$P$ is a graph on 6 vertices~$\{ v_1, v_h, a_1, a_2, b_1, b_2 \}$ and 10 arcs~$\{v_1 a_1,\allowbreak a_1 v_1,\allowbreak v_h a_2,\allowbreak a_2 v_h,\allowbreak a_1 b_1,\allowbreak a_2 b_2,\allowbreak b_1 a_2,\allowbreak b_2 a_1,\allowbreak a_1 b_2,\allowbreak a_2 b_1\}$ with weights defined as follows:
  \begin{align*}
    w(v_1 a_1) &:= 0,     &w(a_1 v_1) &:= 0,     &w(v_{h}a_2) &:= 0, &w(a_2v_{h}) &:= 0, \\
    w(a_1 b_1) &:= C_{R}, &w(a_2 b_2) &:= C_{L}, &w(b_1a_2) &:= 0,     &w(b_2a_1) &:= 0,
  \end{align*}
  \begin{align*}
    w(a_1 b_2) &:=
    \begin{cases}
      C_R & \text{if } \, C_R \leq C_N \text{ or } C_L \leq C_N, \\
      \left\lceil\displaystyle\frac{1}{2} C_N \right\rceil & \textrm{otherwise},
    \end{cases} 
    \\
    w(a_2 b_1) &:=
    \begin{cases}
      C_L & \text{if }  C_R \leq C_N \text{ or } C_L \leq C_N,\\
      \left\lfloor\displaystyle\frac{1}{2} C_N \right\rfloor & \textrm{otherwise}. 
    \end{cases}
  \end{align*}
\end{definition}

\begin{observation}
  \label{obs:cost}
  In \Cref{def:repGadget}, it always holds that~$w(a_1 b_2) \le C_R$, $w(a_2 b_1) \le C_L$, and~$w(a_1 b_2) + w(a_2 b_1) \ge C_N$.
\end{observation}

\begin{reductionrule} \label{reduction:deg2paths}
  Let~$P = (v_0, \ldots, v_{h+1})$ be a maximal induced path of~$G$ with~$h \ge 7$.
  Then, remove the vertices~$v_2, \ldots, v_{h-1}$ of~$P$, add a path-gadget for~$P$ with endpoints~$v_1$ and~$v_h$ (see \Cref{fig:repp}), and keep~$d$ unchanged.
\end{reductionrule}

\begin{figure}[t]
    \tikzset{svertex/.style = {shape=circle,draw,inner sep = 2pt}}
    \tikzset{dotvertex/.style = {shape=circle,inner sep = 2pt}}
    \centering
    \begin{tikzpicture}[scale=0.45]
        \node[svertex] (x0) at (-13.5, 0) {};
        \node[svertex] (x1) at (-11.5, 0) {};
        \node[svertex] (x2) at (-9.5, 0) {};
        \node[dotvertex] (dots) at (-8,0) {...};
        \node[svertex] (xhm1) at (-6.5, 0) {};
        \node[svertex] (xh) at (-4.5, 0) {};
        \node[svertex] (xh1) at (-2.5, 0) {};

        \node at (-13.5, 0.8) {$v_0$};
        \node at (-11.5, 0.8) {$v_1$};
        \node at (-9.5, 0.8) {$v_2$};
        \node at (-6.5, -0.8) {$v_{h-1}$};
        \node at (-4.5, -0.8) {$v_{h}$};
        \node at (-2.5, -0.8) {$v_{h+1}$};

        \path[->] (x0) edge [bend left] node {} (x1);
        \path[->] (x1) edge [bend left] node {} (x0);
        \path[->] (x1) edge [bend left] node {} (x2);
        \path[->] (x2) edge [bend left] node {} (x1);
        \path[->] (x2) edge [bend left] node {} (dots);
        \path[->] (dots) edge [bend left] node {} (x2);
        \path[->] (xhm1) edge [bend left] node {} (dots);
        \path[->] (dots) edge [bend left] node {} (xhm1);
        \path[->] (xhm1) edge [bend left] node {} (xh);
        \path[->] (xh) edge [bend left] node {} (xhm1);
        \path[->] (xh) edge [bend left] node {} (xh1);
        \path[->] (xh1) edge [bend left] node {} (xh);

        \node[svertex] (v0) at (0, 0) {};
        \node[svertex] (v1) at (2, 0) {};
        \node[svertex] (a1) at (4, 0) {};
        \node[svertex] (b1) at (6, 1.8) {};
        \node[svertex] (b2) at (6, -1.8) {};
        \node[svertex] (a2) at (8, 0) {};
        \node[svertex] (vh) at (10, 0) {};
        \node[svertex] (vh1) at (12, 0) {};

        \node at (0, 0.8) {$v_0$};
        \node at (2, 0.8) {$v_1$};
        \node at (4, 0.8) {${a_1}$};
        \node at (6, 2.6) {${b_1}$};
        \node at (6, -2.6) {${b_2}$};
        \node at (8, -0.8) {${a_2}$};
        \node at (10, -0.8) {$v_{h}$};
        \node at (12, -0.8) {$v_{h+1}$};

        \path[->] (v0) edge [bend left] node {} (v1);
        \path[->] (v1) edge [bend left] node {} (v0);
        \path[->, very thick] (v1) edge [bend left] node {} (a1);
        \path[->, very thick] (a1) edge [bend left] node {} (v1);
        \path[->] (a1) edge node {} (b1);
        \path[->] (a1) edge [bend left] node {} (b2);
        \path[->, very thick] (b2) edge [bend left] node {} (a1);
        \path[->] (a2) edge [bend left] node {} (b1);
        \path[->, very thick] (b1) edge [bend left] node {} (a2);
        \path[->] (a2) edge node {} (b2);
        \path[->, very thick] (a2) edge [bend left] node {} (vh);
        \path[->, very thick] (vh) edge [bend left] node {} (a2);
        \path[->] (vh) edge [bend left] node {} (vh1);
        \path[->] (vh1) edge [bend left] node {} (vh);
    \end{tikzpicture}
    \caption{
        Illustration of \Cref{reduction:deg2paths}. 
        We replace the inner vertices of a maximal induced path with a path-gadget. 
        Bold arcs denote arcs of weight 0. 
        For the weights of other arcs in the path-gadget, see \Cref{def:repGadget}.
        The value~$d$ remains unchanged, that is, the cost in both instances are the same.}
    \label{fig:repp}
\end{figure}
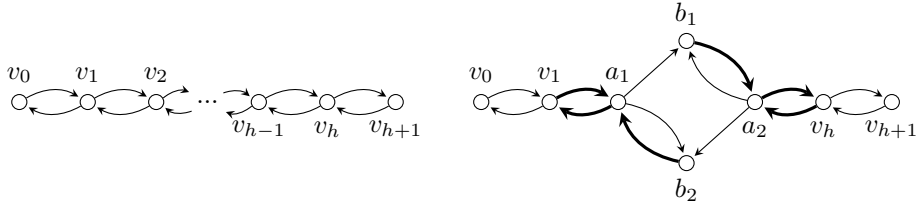

\begin{lemma}
  \label{lemma:deg2paths}
  \Cref{reduction:deg2paths} is correct and can be exhaustively applied in linear time.
\end{lemma}

\begin{proof}
  Let~$I = (G, w)$ be an instance of \textsc{MinPAC} and~$P = ( v_0, \dots, v_{h+1} )$ be a maximal induced path of~$G$ with~$h \ge 7$.
  Let~$I' = (G', w')$ be the instance where the inner vertices of~$P$ are replaced by the path-gadget for~$P$.
  We use~$V_\text{in}$ to denote the inner vertices~$\{ v_1, \dots v_{h} \}$ of~$P$ and~$V_\text{new}$ to denote~$\{ v_1, v_h, a_1, a_2, b_1, b_2 \}$ the new vertices in the path-gadget.
  We show that~$\opt(I) = \opt(I')$.

  $(\ge)$ Let~$H$ be an optimal solution of~$I$.
  Let~$B = A(H) \setminus \{ v_i v_{i+1}, v_{i+1} v_i \mid \allowbreak i \in [h - 1] \}$ be the set of arcs in~$H$ that do have the initial or terminal vertex outside~$V_\text{in}$.
  Let~$B_0 = \{ v_1a_1, a_1v_1, v_ha_2, a_2v_h, b_1a_2, b_2a_1 \}$ be weight-zero arcs inside~$V_\text{new}$ and let~$B_1 = \{ a_1b_1, a_1b_2, a_2b_1, a_2b_2 \}$ be the arcs inside~$V_\text{new}$ that have weight at least one.
  We will construct a solution~$H'$ of~$I'$ such that~$A(H') = B \cup B_0 \cup B_1'$ for some~$B_1' \subseteq B_1$ which we specify later (in a case distinction).
  Thus, in this construction the arcs outside~$V_\text{in}$ and~$V_\text{new}$ remain identical in~$H$ and in~$H'$.
  Hence, it is sufficient to compare the cost for~$V_\text{in}$ in~$H$ and the costs for~$V_\text{new}$ in~$H'$.
  To this end, for~$X \subseteq A(G)$ we define an auxiliary function for the weight of an arc
  \begin{align*}
  \mathbf{w}_{X}(vu) = \begin{cases}
     w(vu) & \text{if } vu \in X, \\
      0 & \text{otherwise}.
     \end{cases}
  \end{align*}
  We define~$\mathbf{w}_{X}'$ analogously for~$w'$ of~$I'$. Using this notation, the cost for~$V_\text{in}$ reads
  \begin{align*}
    \cost(V_\text{in}, A(H), w) 
        ={} & \sum_{i = 1}^{h} \max \{ \mathbf{w}_{A(H)}(v_i v_{i-1}), \mathbf{w}_{A(H)}(v_i v_{i+1}) \} \\
        ={} & \sum_{i = 1}^{h} \mathbf{w}_{A(H)}(v_i v_{i-1}) + \mathbf{w}_{A(H)}(v_{i} v_{i+1}). \\
        ={} & \mathbf{w}_B(v_1v_0) + \mathbf{w}_B(v_hv_{h+1}) \\ & + \sum_{i = 1}^{h - 1} \mathbf{w}_{A(H)}(v_i v_{i+1}) + \mathbf{w}_{A(H)}(v_{i+1} v_{i}).
  \end{align*}
  Here the second equality follows from the assumption that \Cref{reduction:costs} is applied:
  it holds for any~$i \in [h]$ that at least one of~$w(v_i v_{i - 1})$ or~$w(v_{i} v_{i + 1})$ is~0.
  On the other hand, the cost for~$V_\text{new}$ reads
  \begin{align*}
    \cost(V_\text{new}, A(H'), w')
    ={}& \mathbf{w}_{B}(v_1v_0) + \mathbf{w}_{B}(v_{h}v_{h+1}) + \max \{ \mathbf{w}_{B_1'}'(a_1 b_1), \mathbf{w}_{B_1'}'(a_1 b_2) \} \\
    & + \max \{ \mathbf{w}_{B_1'}'(a_2 b_1), \mathbf{w}_{B_1'}'(a_2 b_2) \}. 
  \end{align*}
  We show that~$\delta := \cost(V_\text{in}, A(H), w) - \cost(V_\text{new}, A(H'), w') \ge 0$. 
  We rewrite~$\delta$, by canceling out the terms~$\mathbf{w}_{B}(v_1v_0)$ and~$\mathbf{w}_{B}(v_hv_{h+1})$, as
  \begin{align*}
    \delta &= \cost(V_\text{in}, A(H), w) - \cost(V_\text{new}, A(H'), w') = C_I - C_{I'}
  \end{align*}
  where~$C_I$ and~$C_{I'}$ are given by
  \begin{align*}
    C_I &=  \sum^{h-1}_{i=1} \mathbf{w}_{A(H)}(v_i v_{i+1}) + \mathbf{w}_{A(H)}(v_{i+1} v_i) \text{ and } \\
    C_{I'} &= \max \{ \mathbf{w}_{B_1'}'(a_1 b_1), \mathbf{w}_{B_1'}'(a_1 b_2) \}
    + \max \{ \mathbf{w}_{B_1'}'(a_2 b_1), \mathbf{w}_{B_1'}'(a_2 b_2) \}.
  \end{align*}
  We distinguish between the four cases shown in \Cref{obs:cases}.

  {Case (R).}
  We set~$B_1' := \{ a_1b_1, a_1b_2 \}$.
  Then,~$H'$ is a solution because the connectivity from~$v_1$ to~$v_h$ inside the gadget is preserved.
  Since~$H$ contains arcs~$v_i v_{i+1}$ for all~$i \in [h - 1]$, we have~$C_I \ge C_R$.
  As noted in \Cref{obs:cost}, we have~$C_{I'} = \max \{ w(a_1b_1),  w(a_1b_2) \} \le C_R \le C_I$.

  {Case (L).}
  We set~$B_1' := \{ a_2b_1, a_2b_2 \}$.
  Then, $H'$ is a solution because the connectivity from~$v_h$ to~$v_1$ inside the gadget is preserved.
  Since~$H$ contains arcs~$v_{i+1} v_i$ for all~$i \in [h - 1]$, we have~$C_I \ge C_L$.
  As noted in \Cref{obs:cost}, we have~$C_{I'} = \max \{ w(a_2b_1), w(a_2b_2 )\} \le C_L \le C_I$.

  {Case (B).}
  We set~$B_1' = \{ a_1b_1, a_2b_2 \}$.
  Then,~$H'$ is a solution because the connectivities from~$v_1$ to~$v_h$ and from~$v_h$ to~$v_1$ are both preserved.
  Since~$H$ contains arcs~$v_i v_{i+1}$ and~$v_{i + 1} v_i$ for all~$i \in [h - 1]$, we have~$C_I \ge C_R + C_L$.
  We also have~$C_I' = C_R + C_L$ because the cost for~$a_1$ and~$a_2$ are~$C_R$ and~$C_L$, respectively.

  {Case (N).}
  Here we have~$C_I = C_N$.
  We further distinguish three subcases.
  If~$C_R \le C_N$, then we set~$B_1' := \{ a_1b_1, a_1b_2 \}$.
  Then, $H'$ is a solution with~$C_{I'} = C_R \le C_N$.
  If~$C_L \le C_N$, then we set~$B_1' := \{ a_2b_1, a_2b_2 \}$.
  Then, $H'$ is a solution with~$C_{I'} = C_L \le C_N$.
  Otherwise we set~$B_1' :=\{ a_1 b_2, a_2 b_1 \}$.
  Then, $H'$ is a solution with~$C_{I'} = C_N$.

  $(\le)$ Let~$H'$ now be an optimal solution of~$I'$.
  We can assume that~$H'$ contains all weight-zero arcs in~$B_0$ and some of the non-zero-weight arcs in~$B_1$.
  Let~$B = A(H') \setminus (B_0 \cup B_1)$.
  We construct a solution~$H$ such that~$A(H) \supseteq B$.
  We will give~$H$ by specifying~$B' = A(H) \setminus B$.
  By the same argument as before we compare the following two quantities:
  \begin{align*}
    C_I &=  \sum^{h-1}_{i=1} \mathbf{w}_{B'}(v_i v_{i+1}) + \mathbf{w}_{B'}(v_{i+1} v_i) \text{ and } \\
    C_{I'} &= \max \{ \mathbf{w}_{A(H')}'(a_1 b_1), \mathbf{w}_{A(H')}'(a_1 b_2) \}
    + \max \{ \mathbf{w}_{A(H')}'(a_2 b_1), \mathbf{w}_{A(H')}'(a_2 b_2) \}.
  \end{align*}

  Case 1.
  If there is a path from~$v_1$ to~$v_h$ but no path from~$v_h$ to~$v_1$ inside the path-gadget, then we set~$B' = \{ v_i v_{i+1} \mid i \in [h - 1] \}$.
  Then, $H$ is a solution with~$C_I = C_R$.
  Since~$H'$ contains~$a_1b_1$, it holds that~$C_{I'} \ge C_R$.

  Case 2.
  If there is a path from~$v_h$ to~$v_1$ but no path from~$v_1$ to~$v_h$ inside the path-gadget, then we set~$B' = \{ v_{i+1} v_{i} \mid i \in [h - 1] \}$.
  Then, $H$ is a solution with~$C_I = C_L$.
  Since~$H'$ contains~$a_2 b_2$, it holds that~$C_{I'} \ge C_L$.

  Case 3.
  If there is paths from~$v_1$ to~$v_h$ and from~$v_h$ to~$v_1$ inside the path-gadget, then we set~$B' = \{ v_{i} v_{i+1}, v_{i+1} v_{i} \mid i \in [h - 1] \}$.
  Then, $H$ is a solution with~$C_I = C_R + C_L$.
  Since~$H'$ contains~$a_1 b_1$ and~$a_2 b_2$, it holds that~$C_{I'} \ge C_R + C_L$.

  Case 4.
  If there is neither a path from~$v_1$ to~$v_h$ nor from~$v_h$ to~$v_1$ inside the path-gadget, then we set~$B' = \{ v_{i} v_{i+1}, v_{i+1} v_{i} \mid i \in [h - 1], i \ne k \}$, where~$k$ is~$\arg\max_{i \in [h - 1]} w(v_i v_{i+1}) + w(v_{i+1} v_i)$.
  Then, $H$ is a solution with~$C_I = C_N$.
  Since~$H'$ neither contains~$a_1 b_1$ nor~$a_2 b_2$ in this case, it must contain~$a_1 b_2$ and~$a_2 b_2$ so that~$b_1$ and~$b_2$ can be reached in~$H'$.
  As noted in \Cref{obs:cost}, we have~$C_{I'} \ge w(a_1 b_2) + w(a_2 b_1) \ge C_N$.

  \emph{Running time} 
  To find maximal induced paths, we start with a degree-two vertex and traverse in both directions until a vertex with degree at least 3 is discovered.
  If the maximal induced path contains at least 7 inner vertices, then we replace it with a gadget with appropriate weights.
  The algorithm spends a constant time for each inner vertex in the maximal induced path.
  Since inner vertices of maximal induced paths are pairwise disjoint, this procedure applies \Cref{reduction:deg2paths} exhaustively in linear time.
\end{proof}

\begin{remark}
    \label{rem:LargeRing}
    \Cref{reduction:deg2paths} cannot be applied when~$\uug{G}$ is a large cycle because there is no vertex with degree 3 or larger.
    However, if~$\uug{G}$ is a cycle, then we can easily compute a solution: 
    Let~$v \in V(G)$ be an arbitrary vertex.
    Then, compute costs corresponding to the cases (R), (L), and (N) with~$v_1 = v_h = v$ (see \Cref{fig:cases}).
    Take the cheapest solution found. 
\end{remark}

We have so far shown a reduction rule to remove all degree-one vertices and a gadget to replace every maximal induced path with a fixed number of vertices.
As shown in previous work~\cite{BBNN17}, this is sufficient to obtain a linear-size kernel.

\begin{proposition}[\cite{BBNN17}]
  \label{lem:feedbVertices}
  Any undirected graph~$G$ without degree-one vertices contains at most~$2g-2$ vertices of degree at least three, where~$g$ is the feedback edge number of~$G$.
\end{proposition}

\begin{proposition}[\cite{BBNN17}]
  \label{lem:feedbDeg2Paths}
  Any connected undirected graph~$G$ without degree-one vertices consists of at most~$3g-3$ maximal induced paths, where~$g \ge 2$ is the feedback edge number of~$G$.
\end{proposition}

We use the two propositions above to prove the main theorem of this section.

\begin{proof}[Proof of \Cref{theorem:kernelization}]
  Let~$I = (G, w)$ be an instance of \textsc{MinPAC} with feedback edge number~$g$.
  We apply \Cref{reduction:costs,reduction:degree1} exhaustively to obtain~$I' = (G', w')$, in which there is no degree-one vertex.
  We then obtain~$I'' = (G'', w'')$, in which the inner vertices of each maximal induced path is replaced with a path-gadget using \Cref{reduction:deg2paths}.
  It follows from \Cref{lemma:cost,lemma:degree1correct,lemma:deg2paths} that this transformation is correct and can be done in linear time.

  We show that~$G''$ has at most~$20g-20$ vertices and~$42g-42$ arcs.
  It follows from \Cref{lem:feedbVertices,lem:feedbDeg2Paths} that there are at most~$2g - 2$ vertices of degree at least three and~$3g - 3$ maximal induced paths in~$\uug{G'}$.
  After the exhaustive application of~$\Cref{reduction:deg2paths}$, each maximal induced path~$(v_0, \dots, v_{h+1})$ contains at most 6 inner vertices and 14 arcs (including~$v_0 v_1, v_1 v_0, v_h v_{h+1}, v_{h+1} v_h$).
  Thus, $G''$ contains at most~$2g - 2+ 6 \cdot (3g - 3)=20g-20$ vertices and at most~$14 \cdot (3g - 3) = 42g - 42$ arcs. 
  Note that we count edges between vertices of degree at least three as a maximal induced paths with no inner vertex. 
\end{proof}
}

We can finally again use \Cref{thm:weights} to bound the weights and hence arrive at the following result.

\begin{corollary}
    \label{cor:poly-kernel-fes}
    \textsc{MinPAC} admits a polynomial-size kernel with respect to the feedback edge number.
\end{corollary}

\section{Parameterized hardness} \label{section:hardness}

In this section we present several hardness results for \textsc{MinPAC}.
To this end, we consider the decision variant of \textsc{MinPAC}.

\decProb{\textsc{Power Asymmetric Connectivity (PAC)}}
{A strongly connected graph~$G$, arc weights~$w\colon A(V) \to \N$, and a budget~$k \in \N$.}
{Is there a strongly connected spanning subgraph~$H$ of~$G$, such that~$\cost(V(G), A(H), w) \leq k$?}

We prove that \textsc{PAC} remains NP-hard even if the \emph{feedback arc number} is~1.
This complements the result in \Cref{sec:fen}, where we showed that \textsc{MinPAC} parameterized by the feedback edge number admits an FPT algorithm via a kernelization.
Recall that the feedback arc number for a directed graph is the minimum number of arcs that have to be removed to make it a directed acyclic graph.
Furthermore, we show that \textsc{PAC} is W[2]-hard with respect to the solution cost~$k$.
We also show that \textsc{PAC} cannot be solved in subexponential time in the number of vertices assuming the Exponential Time Hypothesis (ETH)~\cite{Impagliazzo2001}, which states that \textsc{3-Sat} cannot be solved in~$2^{o(n+m)}$ time, where~$n$ and~$m$ are the number of variables and clauses in the input formula.
Summarizing we show the following.

\begin{theorem} \label{theo:hardness} 
  Even if each arc weight is either one or zero and the feedback arc number is 1,
  \begin{enumerate}[align=left,label={(\roman*)}]
    \setlength{\itemindent}{.5em}
    \setlength{\labelwidth}{1.5em}
    \item \textsc{PAC} is NP-hard, \label{hardness:2} %
    \item \textsc{PAC} is W[2]-hard when parameterized by the solution cost~$k$, and \label{hardness:1}
    \item \textsc{PAC} is not solvable in~$2^{o(n)}$ time, unless the ETH fails. \label{hardness:3}
  \end{enumerate}
\end{theorem}

It follows from \Cref{theo:hardness} \ref{hardness:1} that there (presumably) is no algorithm solving \textsc{PAC} running in~$f(k) \cdot n^{O(1)}$ time.
Nonetheless, a simple brute-force algorithm solves \textsc{PAC} in~$n^{\theta(k)}$ time, certifying that \textsc{PAC} is in the class XP with respect to the parameter solution cost.
In order to prove the claims of \Cref{theo:hardness}, we use a reduction from the well-studied \textsc{Set Cover} problem.

\decProb{\textsc{Set Cover}}
{A universe~$U = \{u_1, \ldots , u_n\}$, a set family~$ \mathcal{F} = \{ S_1, \dots , S_m \}$ containing sets~$S_i \subseteq U$, and~$\ell \in \N$.}
{Is there a size-$\ell$ \emph{set cover}~$\mathcal{F}' \subseteq \mathcal{F}$ (that is, $\bigcup_{S \in \mathcal{F}'} S = U$)?}

\textsc{Set Cover} is NP-hard and W[2]-hard with respect to the solution size~$\ell$~\cite{downey2013} and is not solvable in~$2^{o( \vert \mathcal{U} \vert + \vert \mathcal{F} \vert)}$ time unless the ETH fails \cite{paturi2001}.

For the reduction, we use one vertex for each element and each subset and one arc to represent the membership of an element in a subset.
The construction resembles the one used in \textsc{Min-Power Symmetric Connectivity}~\cite{BBNN17}.

\begin{reduction} \label{red1}
  Given an instance~$I = (U, \mathcal{F}, \ell)$ of \textsc{Set Cover}, we construct an instance~$I' = (G, w, k = \ell)$ of \textsc{PAC} as follows.
  We introduce a vertex~$v_u$ for every~$u \in U$, a vertex~$v_S$ for every~$S \in \mathcal{F}$, and two additional vertices~$s$ and~$t$.
  We construct a graph such that~$V(G) = \{ s, t \} \cup V_U \cup V_\mathcal{F}$ where~$V_U = \{ v_u \mid u \in U\}$ and~$V_\mathcal{F} = \{ v_S \mid S \in \mathcal{F} \}$.
  For the arcs we first add an arc~$ts$ of weight 0.
  We then add arcs~$s v_S$ and~$v_S t$ of weight 0 for every~$S \in \mathcal{F}$ and an arc~$v_u t$ of weight 0 for every~$u \in U$.
  For every~$S \in \mathcal{F}$ and every~$u \in S$ we finally add an arc~$v_S v_u$ of weight 1.
\end{reduction}

\Cref{fig:red} illustrates the reduction to \textsc{PAC}.
We can assume that arcs of weight zero (bold arcs in the figure) are part of the solution.
The idea is that in order to obtain a strongly connected subgraph, one has to select at least one incoming arc for each vertex in~$V_U$ such that only~$k$ vertices in~$V_{\mathcal{F}}$ have outgoing arcs that are selected.

\begin{figure}[t] 
  \centering
  \tikzset{svertex/.style = {shape=circle,draw,inner sep = 2pt}}
  \tikzset{dotvertex/.style = {shape=circle,inner sep = 2pt}}
  \begin{tikzpicture}[scale=0.8]

    \node[svertex] (s) at (0, 0) {};

    \node[svertex] (S1) at (1.5, 0) {};
    \node[svertex] (S2) at (2.5, 0) {};

    \node[svertex] (U1) at (4, 0) {};
    \node[svertex] (U2) at (5, 0) {};
    \node[svertex] (U3) at (6, 0) {};

    \node[svertex] (t) at (7.5, 0) {};

    \node at (0, 0.5) {$s$};
    \node at (7.5, 0.5) {$t$};

    \path[->, very thick] (t) edge [bend left=100] node {} (s);
    \path[->, very thick] (s) edge node {} (S1);
    \path[->, very thick] (s) edge [bend right=25] node {} (S2);
    \path[->] (S2) edge node {} (U1);
    \path[->] (S2) edge [bend left=25] node {} (U2);
    \path[->] (S1) edge [bend left=50] node {} (U2);
    \path[->] (S1) edge [bend left=75] node {} (U3);
    \path[->, very thick] (S1) edge [bend right=80] node {} (t);
    \path[->, very thick] (S2) edge [bend right=60] node {} (t);
    \path[->, very thick] (U1) edge [bend right=40] node {} (t);
    \path[->, very thick] (U2) edge [bend right=20] node {} (t);
    \path[->, very thick] (U3) edge node {} (t);

    \begin{scope}[on background layer]
      \node[draw = black, dashed, rectangle, rounded corners, minimum height = 0.8cm, minimum width = 1.5cm] at (2, 0){};
      \node[draw = black, dashed, rectangle, rounded corners, minimum height = 0.8cm, minimum width = 2.5cm] at (5, 0){};  
      \node[text width=1cm] at (1.5, 1) {$V_\mathcal{F}$};
      \node[text width=1cm] at (7, 1) {$V_U$};
    \end{scope}
  \end{tikzpicture}
    \caption{Illustration of \Cref{red1} on a \textsc{Set Cover} instance with universe~$U = \{ 1, 2, 3 \}$ and set family~$\mathcal{F} = \{ \{ 2 ,3 \}, \{ 1, 2\} \}$. Bold arcs denote arcs of weight 0 and other arcs have weight 1. }
    \label{fig:red}
\end{figure}
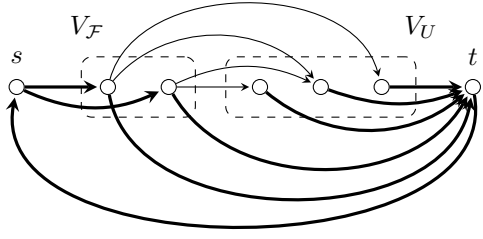

{
To prove the W[2]-hardness, we have to verify that the given reduction is indeed a \emph{parameterized reduction}.

\begin{definition} \label{Pred}
  A \emph{parameterized reduction} from a parameterized problem~$\Pi \subseteq \Sigma^* \times \Sigma^*$ to a parameterized problem~$\Pi' \subseteq \Sigma^* \times \Sigma^*$ is a function which maps any instance~$(I,p) \in \Sigma^* \times \Sigma^*$ to another instance~$(I',p')$ such that
  \begin{enumerate}[align=left,label={(\roman*)}]
    \setlength{\itemindent}{.5em}
    \setlength{\labelwidth}{1.7em}
    \item $(I',p')$ can be computed from~$(I,p)$ in~$f(p) \cdot \vert I \vert ^{O(1)}$ time for some computable function~$f$, \label{Pred:1}
    \item $p' \leq g(p)$ for some computable function~$g$, and \label{Pred:2}
    \item $(I,p) \in \Pi$ if and only if~$(I',p') \in \Pi'$. \label{Pred:3}
  \end{enumerate}
\end{definition}

\begin{lemma} \label{hardness:parameterreduction}
  \Cref{red1} is a parameterized reduction from \textsc{Set Cover} parameterized by the solution size to \textsc{PAC} parameterized by the solution cost.
\end{lemma}

\begin{proof}
  To prove that \Cref{red1} is a parameterized reduction, we verify \Cref{Pred}~\ref{Pred:1} to \ref{Pred:3}.
  Observe that \Cref{red1} can be done in~$O(|U| + |\mathcal{F}|)$ time, which satisfies \Cref{Pred}~\ref{Pred:1}.
  \Cref{Pred}~\ref{Pred:2} is clearly satisfied.
  For \Cref{Pred}~\ref{Pred:3}, we show that~$I$ has a set cover of size at most~$\ell$ if and only if~$G$ has a strongly connected subgraph~$H$ of cost is at most~$\ell$.

  ($\Rightarrow$)
  Let~$\mathcal{F'} \subseteq \mathcal{F}$ be a set cover of size at most~$\ell$.
  Let~$B_0 = \{ s v_S, v_S t \mid S \in \mathcal{F} \} \cup \{ v_u t \mid u \in U \}$ be the arcs of weight 0.
  We claim that~$H = (V(G), B_0 \cup \{ v_{S'} v_{u'} \mid S' \in \mathcal{F}', u' \in S' \})$ is a solution with cost~$\ell$.
  Since~$\mathcal{F}'$ is a set cover, there exists at least one incoming arc in~$H$ for any vertex in~$V_U$.
  Thus, $H$ is strongly connected.
  Since the cost for~$v_{S'}$ is 1 for any~$S' \in \mathcal{F}'$ and the costs for other vertices are 0, the cost of~$H$ is at most~$\ell$.

  ($\Leftarrow$)
  Let~$H$ be a strongly connected subgraph of cost at most~$\ell$.
  Let~$\mathcal{F}' = \{ S \mid \exists u.\; v_S v_u \in A(H) \}$.
  Then, $\mathcal{F}'$ is a set cover because there is at least one incoming arc in~$H$ for any~$v_u \in V_U$.
  Since the cost of~$H$ is at most~$\ell$, we have~$|\mathcal{F}'| \le \ell$.
\end{proof}

Now we can prove the statements of the theorem.

\begin{proof}[Proof of \Cref{theo:hardness}]
  \Cref{theo:hardness}~\ref{hardness:1} follows from \Cref{hardness:parameterreduction} because \textsc{Set Cover} is W[2]-hard when parameterized by the solution size \cite{downey2013}.
  For \Cref{theo:hardness}~\ref{hardness:2}, observe that \Cref{red1} is a polynomial-time reduction from \textsc{Set Cover} and the constructed graph has a feedback arc set of size 1.
  For \Cref{theo:hardness}~\ref{hardness:3}, observe that the constructed graph of \Cref{red1} has~$O(|U| + |\mathcal{F}|)$ vertices. Since \textsc{Set Cover} cannot be solved in~$2^{o(\vert U \vert + \vert \mathcal{F} \vert)}$ time assuming ETH \cite{paturi2001}, \Cref{theo:hardness}~\ref{hardness:3} follows.
\end{proof}
\begin{remark}
    We remark that having arcs of weight zero is essential for the W[2]-hardness in \cref{theo:hardness}~\ref{hardness:1}:
    If~$\min_{vu \in A(G)} w(vu) \ge 1$ for any~$v \in V(G)$, then \textsc{PAC} is trivially FPT with respect to the solution cost (as the cost is at least~$n$).
    However, even if~$\min_{vu \in A(G)} w(vu) \ge 1$, \textsc{PAC} is still W[2]-hard with respect to the \emph{above lower bound}~$k - \sum_{v \in V(G)} \min_{vu \in A(G)} w(vu)$ (this follows from a modification to \Cref{red1} where every arc weight is increased by one).
\end{remark}
}

\section{Conclusion}

We started the investigation of the parameterized complexity of \textsc{MinPAC}, leading to first tractability and intractability results.
We remark that our algorithms run in linear time when the respective parameters are bounded.
Thus we believe that our results are worthwhile for empirical experiments. 
There are also several theoretical challenges for future work:
Can the running time of the parameterized algorithm with respect to the number~$c$ of SCCs in the obligatory subgraph be improved to single-exponential?
Resolving the parameterized complexity of \textsc{MinPAC} with respect to the single parameter vertex cover number is another task for future work.
Finally, problem variants where the solution graph is not only required to be strongly connected but needs to have at most a certain diameter might be interesting (theoretically and from an application point of view where the number of hops for communication should be limited).

\bibliography{ref}
\bibliographystyle{plainnat} 

\end{document}